\documentclass[pdflatex,sn-mathphys-num]{sn-jnl}

\usepackage{graphicx}%
\usepackage{multirow}%
\usepackage{amsmath,amssymb,amsfonts}%
\usepackage{amsthm}%
\usepackage{mathrsfs}%
\usepackage[title]{appendix}%
\usepackage{xcolor}%
\usepackage{textcomp}%
\usepackage{manyfoot}%
\usepackage{booktabs}%
\usepackage{algorithm}%
\usepackage{algorithmicx}%
\usepackage{algpseudocode}%
\usepackage{listings}%

\usepackage{stackengine} 

\usepackage{comment}

\theoremstyle{thmstyleone}%
\newtheorem{theorem}{Theorem}
\newtheorem{proposition}[theorem]{Proposition}%

\theoremstyle{thmstyletwo}%
\newtheorem{remark}{Remark}%

\theoremstyle{thmstylethree}%
\newtheorem{definition}{Definition}%

\newcommand{\R}{\mbox{$ \mathbb{R}  $}}

\newcommand{\pmat}[1]{\begin{pmatrix} #1 \end{pmatrix}}

\newcommand{\vs}{\mathbf{v}_\mathbf{s}}
\newcommand{\ve}{\mathbf{v}_\mathbf{e}}

\newcommand{\s}{\mathbf{s}}

\newcommand{\ee}{\mathbf{e}}

\newtheorem{corollary}{\bf Corollary}[section]

\newtheorem{case}{\bf Case}

\newcommand{\gplus}{\ensurestackMath{\stackengine{0pt}{\, \circlearrowleft\,}{+}{O}{c}{F}{F}{L}}}  

\newcommand{\bd}{\begin{definition}}
	\newcommand{\ed}{\end{definition}}
\newcommand{\bt}{\begin{theorem}}
	\newcommand{\et}{\end{theorem}}
\newcommand{\bi}{\begin{itemize}}
	\newcommand{\ei}{\end{itemize}}
\newcommand{\ben}{\begin{enumerate}}
	\newcommand{\een}{\end{enumerate}}
\newcommand{\beq}{\begin{equation}}
	\newcommand{\eeq}{\end{equation}}

\begin{document}

\title[Article Title]{Markovian dynamics of single-rebit open quantum systems with applications to colour perception \\}


\author[1]{\sur{Michel Berthier}}\email{michel.berthier@univ-lr.fr}

\author[2,3]{\sur{Gabriel Niebel}}\email{gabriel.niebel@math.u-bordeaux.fr}

\author*[2]{\sur{Edoardo Provenzi}}\email{edoardo.provenzi@math.u-bordeaux.fr}

\affil[1]{Laboratoire MIA, Batiment Pascal, P\^ole Sciences et Technologie, Universit\'e de La Rochelle, 23, Avenue A. Einstein, BP 33060, 17031\\ La Rochelle cedex, France}

\affil[2]{Universit\'e de Bordeaux, CNRS, Bordeaux INP, IMB, UMR 5251, F-33400, 351 Cours de la Lib\'eration, Talence, France}

\affil[3]{also with Huawei SASU France}


\abstract{This paper investigates the Markovian dynamics of open two-state quantum systems defined over the real numbers (rebits). Two main objectives are pursued. First, we present a comprehensive classification of Markovian rebit quantum channels, i.e. one-parameter semigroups of completely positive, trace-preserving (CPTP) maps acting on the rebit state space. We show that a full characterisation of their action can be achieved and that describing these channels as solutions of the GKSL equation allows us to explicitly identify the associated Lindblad generators and conditions for complete positivity. Second, we present an original application of this classification to colour perception. Using a recent model in which perceived colours arise from L\"uders measurements on the rebit state space, we show how chromatic distortion induced by a non-neutral illuminant can be modelled by a Markovian rebit channel that progressively diminishes colour distinguishability. Other types of channels could be used to study colour vision deficiencies. These phenomena are illustrated by simulations on digital images, highlighting the relevance of rebit Markovian dynamics in modelling colour vision.}




\maketitle

\section{Introduction}
\label{sec1}
The study of open quantum systems has led to the development of a rich mathematical framework for describing the dynamical evolution of quantum states, among which Markovian dynamics plays a central role. In standard quantum information theory, such studies are carried out in the complex context, where qubits (two-level quantum systems) evolve under completely positive, trace-preserving (CPTP) linear maps, hereafter referred to as channels, and are governed by master equations such as the Gorini-Kossakowski-Sudarshan-Lindblad (GKSL) equation. 

However, recent years have seen a growing interest in studying real-valued quantum systems, where the qubit is replaced by its analogue over the real numbers, known as the rebit. Several works, including \cite{Wootters:1990,McKague:2009,Aleksandrova:13,Wootters:2014,Moretti:2017quantum,Koh:2018,Renou:2021,Fuchs:2022,Chiribella:2023}, provide insightful theoretical contributions to the understanding of rebit dynamics. It is worth noting that the rebit is not a simple by-product of the qubit, and that there are dramatic differences between the real and complex settings. One of these differences is illustrated by the fact that, when $n\geq2$, the dimension of the space of $d^n\times d^n$ real symmetric matrices is strictly greater than the $n$th power of the dimension of the space of $d\times d$ real symmetric matrices \cite{Caves:2002}: $d^n(d^n + 1)/2 > (d(d + 1)/2)^n$, whereas these dimensions agree for Hermitian matrices in the complex settings. According to \cite{Wootters:1990}, in real vector space quantum mechanics, one cannot in general determine the density matrix using only local measurements. 

Another difference comes from the fact that entanglement over the real numbers differs from that of standard quantum mechanics over the complex numbers. For instance, the real state 
$
\rho=\left (I\otimes I+\sigma_y\otimes\sigma_y\right)/4
$
is complex separable because it can be written as
$
\rho=\left(\frac{1}{2}\left(I+\sigma_y\right)\otimes\frac{1}{2}\left(I+\sigma_y\right)+\frac{1}{2}\left(I-\sigma_y\right)\otimes\frac{1}{2}\left(I-\sigma_y\right)\right)/2
$ in the complex framework, but real entangled because its concurrence is equal to 1, see \cite{Caves:2001}.

In the first, purely theoretical part of this work, we derive a comprehensive classification of Markovian rebit channels. We begin by analysing the case in which the affine component of the dynamics is itself Markovian, for which the classification admits a relatively simple and explicit form, and then extend the analysis to the general situation. By establishing a precise connection with the GKSL master equation adapted to rebit systems, we are able to formulate clear conditions ensuring complete positivity of the resulting dynamical maps.

The second part of the paper is devoted to an application to a recent quantum information-based model of colour perception, in which chromatic states are represented by rebit states. We show that chromatic distortion arising from illuminant changes can be coherently modelled by non-unital Markovian rebit channels acting on states. Such dynamics naturally accounts for the progressive reduction of chromatic distinguishability when human observers are exposed to non-neutral illuminants. Other types of channels could be used to model colour vision deficiencies.

The paper is structured as follows. For clarity of exposition, the complete classification of Markovian rebit channels is developed in separate subsections of Section \ref{sec:Markovclassification}. In \ref{subsec:generalclassification}, we begin by introducing the general setting of rebit systems and recall the basic notions concerning rebit states, quantum channels, and Markovian dynamics. Subsection \ref{subsec:classmarkovian^2} is devoted to the classification of Markovian rebit channels whose affine component is itself Markovian. The technically more involved case of general Markovian rebit channels is addressed by first considering unital Markovian dynamics in  \ref{subsec:unitalgeneral}, which provides the key structural insight leading naturally to the full classification presented in \ref{sec:general}.

In Section \ref{sec:GKSL}, we establish a direct connection between Markovian rebit channels and the GKSL master equation adapted to real two-level quantum systems, deriving explicit and operational conditions for complete positivity in terms of the coefficients of the Lindblad matrix.

Finally, in Section \ref{sec:applications}, we apply these results to a quantum information-based model of colour perception, in which chromatic states are represented as rebit states. We show how illuminant-induced chromatic distortion can be coherently modelled by suitable non-unital Markovian rebit channels, and illustrate the resulting dynamics through numerical simulations.

\section{Classification of Markovian rebit channels}\label{sec:Markovclassification}

Markovian rebit channels form a subclass of general rebit channels, whose classification has been established in \cite{Alde:2023}. The main results of that work which are relevant for the present study are briefly recalled in the following subsection, in which we also define precisely the concept of Markovian behaviour for rebit channels.

\subsection{Classification of general rebit channels and Markovian behaviour}\label{subsec:generalclassification}

The rebit is the two-state quantum system with Hilbert space $\mathbb R^2$. The density matrices $\rho_{\bf s}$ representing rebit states form a convex subset $\mathcal{S}(\mathcal H(2,\R))$ of $\mathcal H(2,\R)$, the 3-dimensional $\R$-vector space of $2\times 2$ real symmetric matrices. Every $\rho_{\bf s}$ is a unit-trace, positive semi-definite matrix which can be explicitly expressed as:
\beq\label{eq:densmatrebit}
\rho_{\bf s}=\frac{1}{2}\pmat{1+s_1 & s_2 \\ s_2 & 1-s_1}=\frac 1 2 ( \sigma_0 + s_1 \sigma_1 + s_2 \sigma_2),
\eeq 
where $s_1^2+s_2^2\le 1$, $\sigma_0$ is the identity matrix $I_2$ and $\sigma_1,\sigma_2$ are the Pauli matrices with real entries, 
\begin{equation}\label{eq:Paulireal}
	\sigma_1=\begin{pmatrix}
		1 & 0 \\ 0 & -1
	\end{pmatrix}, \; \sigma_2=\begin{pmatrix}
		0 & 1 \\ 1 & 0
	\end{pmatrix}.
\end{equation} 
The density matrix $\rho_{\bf s}$ is univocally associated to a vector $\vs=(s_1,s_2)$ belonging to the unit disk $\cal D$ in $\R^2$. In this context, $\vs$ and $\cal D$ are referred to as \textit{Bloch vector} and \textit{Bloch disk}, respectively. 
Throughout this paper, a rebit state $\s$ will be equivalently identified with its density matrix $\rho_{\s}$ or its associated Bloch vector $\vs$, depending on the context.

It is shown in \cite{Alde:2023} that the quantum channels, i.e. the CPTP linear maps, of the rebit transform the Bloch disk through the following maps
\begin{equation}
	\mathcal C(\vs)=R_1{\mathcal C}_{\mathbb A}R_2 \vs=R_1 D R_2\vs+R_1{\bf w},
\end{equation}
where $R_i$ is a rotation matrix with angle $\theta_i\in [0,2\pi)$, $i=1,2$, and the affine component ${\mathcal C}_{\mathbb A}$ of ${\mathcal C}$ acts as follows
\begin{equation}\label{eq:affineaction}
	{\mathcal C}_{\mathbb A}(\vs)=D\vs+{\bf w},
\end{equation} 
with $D=\text{diag}(\lambda_1,\lambda_2)$ and ${\bf w}=(w_1,w_2)\in \cal D$. The complete positivity (CP) of these maps is expressed by the following conditions
\begin{equation}
	\label{eq:channelconstraints}
	\begin{cases}
		1+\lambda_1+\lambda_2 \geq 0\\[1ex]
		1+\lambda_1-\lambda_2 \geq 0\\[1ex]
		1-\lambda_1+\lambda_2 \geq 0\\[1ex]
		\displaystyle\frac{w_1^2}{(1+\lambda_1+\lambda_2)(1+\lambda_1-\lambda_2)}+\frac{w_2^2}{(1+\lambda_1+\lambda_2)(1-\lambda_1+\lambda_2)}\leq 1,
	\end{cases}
\end{equation}
together with the requirement that if $\max\{\lambda_1,\lambda_2\}=1$, then ${\bf w}=\bf{0}$ in order to preserve $\cal D$. 

The first three constraints imply that $\lambda_1$ and $\lambda_2$ are confined to a subset $\cal A$ of the square defined by $|\lambda_i|\le 1$, $i=1,2$. $\cal A$ is called \emph{admissibility region} and is illustrated in Figure \ref{fig:adm_region}.
\begin{figure}[!ht]
	\centering
	\includegraphics[scale=0.8]{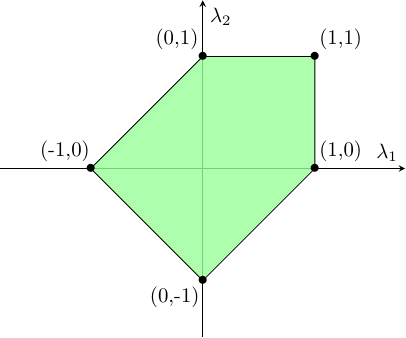}
	\caption{Admissibility region $\cal A$ for the parameters $\lambda_1,\lambda_2$.}
	\label{fig:adm_region}
\end{figure}

The action of $C_{\mathbb A}$ on the Bloch disk $\mathcal D$ is twofold. First, $\mathcal D$ is contracted into a two-dimensional ellipsoidal region with semiaxes $\lambda_1$ and $\lambda_2$. Second, the centre of $\mathcal D$ is translated by the vector $\mathbf w$. Accordingly, $\lambda_1$ and $\lambda_2$ are referred to as the \emph{scaling coefficients}, while $\mathbf w$ is called the \emph{shift vector}. A rebit channel is \emph{unital}, that is, it preserves the centre of $\mathcal D$, if and only if $\mathbf w=\mathbf 0$.

We now define the rebit channels that will be analysed in this paper.

\begin{definition}
	A Markovian rebit channel is a continuous one-parameter semigroup $\mathcal C(t)$, $t\geq 0$, of rebit channels. More precisely, the family of quantum channels
	\begin{equation}\label{eq:gen}
		\mathcal C(t)(\vs)=R_1(t)\,\mathcal C_{\mathbb A}(t)\,R_2(t)\vs
		=R_1(t)D(t)R_2(t)\vs+R_1(t)\mathbf w(t), \qquad t\geq 0,
	\end{equation}
	with
	\begin{equation}
		D(t)=\mathrm{diag}(\lambda_1(t),\lambda_2(t)), \qquad 
		\mathbf w(t)=(w_1(t),w_2(t)), \qquad t\ge 0,
	\end{equation}
	satisfies the semigroup conditions
	\begin{equation}
		\label{eq:rebitmarkov}
		\begin{cases}
			\mathcal C(t_1)\circ\mathcal C(t_2)=\mathcal C(t_1+t_2)\\
			\mathcal C(0)=I_2,
		\end{cases}
	\end{equation}
	for all $t_1,t_2\ge 0$, the map $t\mapsto \mathcal C(t)$ is continuous, and the scaling and shift functions $\lambda_i(t)$ and $w_i(t)$, $i=1,2$, satisfy the complete positivity conditions \eqref{eq:channelconstraints} for every $t\ge 0$.
\end{definition}

For the sake of brevity, in the following we shall simply refer to the CP conditions \eqref{eq:channelconstraints}, with the implicit understanding that they are required to hold for the scaling and the shift functions for all $t\ge 0$. Notice that a Markovian rebit channel ${\cal C}(t)$ is unital if ${\bf w}(t)={\bf 0}$ for all $t\geq 0$.

Thanks to a standard result, see e.g. \cite{Engel:2000}, we have that 
\begin{equation}
	\label{eq:matexpC}
	\mathcal C(t)= \exp\left(\dot {\mathcal C}(0)t\right), \quad t\ge 0,
\end{equation} 
where $\dot {\mathcal C}(0)$ is the infinitesimal generator, 
\begin{equation}
	\dot {\mathcal C}(0)=\left. \frac{d\mathcal C(t)}{dt}\right|_{t=0}.
\end{equation}
In the next subsection, we show that the classification of Markovian rebit channels is quite simple when the affine component
\begin{equation}
	\mathcal{C}_{\mathbb{A}}(t)(\vs) = D(t) \vs + \mathbf{w}(t)
\end{equation}
of the channel $\mathcal{C}(t)$ is itself Markovian.

\subsection{Classification of Markovian rebit channels with a Markovian affine component}
\label{subsec:classmarkovian^2}

The following proposition provides the exhaustive classification of Markovian rebit channels $\mathcal{C}(t)$ whose affine component $\mathcal{C}_{\mathbb{A}}(t)$ is itself Markovian.

\begin{proposition}
	\label{prop1.1}
	A Markovian rebit channel $\mathcal{C}(t)$ whose affine component $\mathcal{C}_{\mathbb{A}}(t)$ is itself Markovian admits one of the following two representations.
	\ben
	\item Given $-a := \dot\lambda_1(0)=\dot\lambda_2(0)\leq 0$ and $\omega := \dot\theta_1(0)+\dot\theta_2(0)$, 
	\begin{equation}\label{eq:affineCaffineCA1}
		\mathcal{C}(t)=\mathcal{C}_{\mathbb{A}}(t)\,R(\omega t)
		=\operatorname{diag}(e^{-a t},e^{-a t})\,R(\omega t),
	\end{equation}
	where $R(\omega t)$ denotes the rotation matrix of angle $\omega t$.
	\medskip
	\item Given $a_i := -\dot\lambda_i(0)\geq 0$ and $b_i := \dot w_i(0)$, $i=1,2$,
	\begin{equation}\label{eq:affineCaffineCA2}
		\mathcal{C}(t)=R_2(0)^{-1}\,\mathcal{C}_{\mathbb{A}}(t)\,R_2(0),
	\end{equation}
	where
	\begin{equation}
		\mathcal{C}_{\mathbb{A}}(t)(\vs)
		=\operatorname{diag}(e^{-a_1 t},e^{-a_2 t}) \vs+\mathbf{w}(t),
	\end{equation}
	and 
	\begin{equation}
		w_i(t)=\frac{b_i}{a_i}\left(1-e^{-a_i t}\right), \quad i=1,2,
	\end{equation}
	are such that the CP conditions \eqref{eq:channelconstraints} are fulfilled for all $t\geq 0$.
	\een
\end{proposition}

\proof
Let $\mathcal C(t)$ be a Markovian channel whose affine component ${\mathcal C}_{\mathbb A}(t)$ is Markovian. The condition 
\beq
{\cal C}_{\mathbb A}(t_1+t_2)
=
{\cal C}_{\mathbb A}(t_1)\circ{\cal C}_{\mathbb A}(t_2),
\quad t_1,t_2\ge 0,
\eeq
together with eq. \eqref{eq:gen} implies 
\beq
D(t_1+t_2)\vs+{\bf w}(t_1+t_2) = D(t_1)D(t_2)\vs+D(t_1){\bf w}(t_2)+{\bf w}(t_1),
\eeq
i.e.
\beq\label{eq:characters}
\lambda_i(t_1+t_2)=\lambda_i(t_1)\lambda_i(t_2)
\eeq 
and
\beq\label{eq:wMarkov}
w_i(t_1+t_2)=w_i(t_1)+\lambda_i(t_1)w_i(t_2),
\eeq
$i=1,2$, for all $t_1,t_2\geq 0$. Moreover, the initial condition ${\cal C}_{\mathbb A}(0)=I_2$ implies 
\begin{equation}\label{eq:incondCA}
	{\mathcal C}_{\mathbb A}(0)=R_1^{-1}(0)R_2^{-1}(0)=I_2,
\end{equation}
so ${\bf w}(0)={\bf 0}$, $\theta_1(0)=-\theta_2(0)$ and, by the CP conditions \eqref{eq:channelconstraints},  $\lambda_1(0)=\lambda_2(0)=1$, so $D(0)=I_2$.

Eq. \eqref{eq:characters} implies that $\lambda_1(t)$ and $\lambda_2(t)$ are continuous characters of the additive semigroup $[0,+\infty)$, so
\begin{equation}
	\lambda_i(t)=e^{\dot\lambda_i(0)t},
\end{equation}
with $\dot \lambda_i(0)\leq 0$, $i=1,2$, to respect the admissibility region. Introducing this result in \eqref{eq:wMarkov} and differentiating with respect to $t_2$, evaluated at $t_2=0$, we obtain
\beq
\dot w_i(t)=\dot w_i(0)\,e^{\dot\lambda_i(0)t},
\eeq
which, considering the initial condition ${\bf w}(0)={\bf 0}$, gives
\begin{equation}\label{eq:wit}
	w_i(t) = \frac{\dot w_i(0)}{\dot \lambda_i(0)} \left( e^{\dot \lambda_i(0) t} -1\right).
\end{equation}
We denote, as in the statement of the proposition, $-a_i=\dot \lambda_i(0)\leq 0$, and $b_i=\dot w_i(0)$, for $i=1,2$. Then, by eq. \eqref{eq:wit}, ${\mathcal C(t)}$ is unital if and only if $b_1=b_2=0$. Let us first analyse such unital rebit channels in the case $a_1=a_2=:a$, then
\begin{equation}
	\mathcal C(t)=R_1(t) D(t) R_2(t) = e^{-a t} R_1(t) R_2(t),
\end{equation}
so
\begin{equation}
	\mathcal C(t_1+t_2)=e^{-a (t_1 + t_2)} R_1(t_1 + t_2) R_2(t_1 + t_2)= e^{-a {t_1}} e^{-a {t_2}} R_1(t_1) R_1(t_2) R_2(t_1) R_2(t_2),
\end{equation}
and $R_1(t_1 + t_2) R_2(t_1 + t_2)= R_1(t_1)R_2(t_1) R_1(t_2)R_2(t_2)$. Consequently,
\begin{equation}
	\theta_1(t_1+t_2)+\theta_2(t_1+t_2)=(\theta_1(t_1)+\theta_2(t_1))+(\theta_1(t_2)+\theta_2(t_2)),
\end{equation}
so, the continuous angular function $\theta(t)=\theta_1(t)+\theta_2(t)$ satisfies the Cauchy functional equation $\theta(t_1 + t_2) = \theta(t_1) + \theta(t_2)$ for all $t_1,t_2\ge 0$, which implies that $\theta(t)=\omega t$, with $\omega=\dot\theta_1(0)+\dot\theta_2(0)$, thus giving the expression \eqref{eq:affineCaffineCA1}. 

Suppose now that $a_1\neq a_2$. By the Markovian property of $\mathcal C_{\mathbb A}$ we have
\begin{equation}
	R_1(t_1) D(t_1) R_2(t_1) R_1(t_2) D(t_2) R_2(t_2) \vs = R_1(t_1 + t_2) D(t_1 + t_2) R_2(t_1 + t_2) \vs,
\end{equation} 
for all $\vs\in\mathcal D$, hence
\begin{equation}
	\|D(t_1) R_2(t_1) R_1(t_2) D(t_2) R_2(t_2) {\vs}\| = \|D(t_1 + t_2) R_2(t_1 + t_2) {\vs}\|.
\end{equation}
The diagonal matrix $D(t)$ has eigenvalues $e^{-a_1t}$ and $e^{-a_2t}$, with eigenvectors of coordinates $(1,0)$ and $(0,1)$, respectively. So, if we consider
\begin{equation}
	{\vs}= R_2^{-1}(t_1 + t_2)\begin{pmatrix} 1\\0\end{pmatrix},
\end{equation}
then we have 
\begin{equation} 
	\|D(t_1 + t_2) R_2(t_1 + t_2) {\vs}\|= e^{-a_1 (t_1 + t_2)},
\end{equation} 
and the equality
\begin{equation}
	\|D(t_1) R_2(t_1) R_1(t_2) D(t_2) R_2(t_2) {\vs}\| = e^{-a_1 (t_1 + t_2)}
\end{equation}
can be fulfilled if and only if
\begin{equation}
	\label{eq:r2sts}
	R_2(t_2) = \pm R_2(t_1 + t_2),
\end{equation}
and
\begin{equation}
	\label{eq:r2r1}
	R_2(t_1) R_1(t_2) = \pm I_2,
\end{equation}
for all $t_1,t_2\geq 0$.
Equation \eqref{eq:r2sts} implies $\theta_2(t_2) = \theta_2(t_1 + t_2)$ or $\theta_2(t_2) = \theta_2(t_1 + t_2) + \pi$, for all $t_1,t_2\ge 0$. This means that $R_2(t)$ is a constant rotation and, by eq. \eqref{eq:r2r1}, $R_2(t)=R_2(0) = \pm R_1^{-1}(t)$. Since $C(0)$ is the identity, $R_1(t) = R_2^{-1}(0)$, and we obtain expression \eqref{eq:affineCaffineCA2} in the unital case.

Suppose now that $\mathcal C(t)$ is non-unital. By Brouwer's fixed point theorem, together with the uniqueness of solutions to ordinary differential equations, ${\mathcal C}_{\mathbb A}(t)$ and ${\mathcal C}(t)$ have a fixed point, that we indicate with ${\bf v}_{{\bf s}_0}$ and ${\bf v}_{{\bf s}_1}$, respectively, both contained in $\mathcal D$ for all $t\geq 0$. From
\beq
{\bf v}_{{\bf s}_0}={\mathcal C}_{\mathbb A}(t){\bf v}_{{\bf s}_0}=D(t){\bf v}_{{\bf s}_0}+{\bf w}(t),
\eeq 
we get
\begin{equation}
	{\bf w}(t)={\bf v}_{{\bf s}_0}-D(t) {\bf v}_{{\bf s}_0},
\end{equation}
so
\begin{equation}
	{\bf v}_{{\bf s}_1}=\mathcal C(t) {\bf v}_{{\bf s}_1} = R_1(t) \left[D(t) R_2(t) {\bf v}_{{\bf s}_1} + {\bf v}_{{\bf s}_0} - D(t) {\bf v}_{{\bf s}_0}\right],
\end{equation}
or, equivalently
\begin{equation}
	\label{eq:rot}
	D(t) (R_2(t) {\bf v}_{{\bf s}_1} - {\bf v}_{{\bf s}_0}) = R_1^{-1}(t) {\bf v}_{{\bf s}_1} - {\bf v}_{{\bf s}_0}, \quad t\geq 0.
\end{equation}
Since $D(t)$ is invertible, the left-hand side of \eqref{eq:rot} vanishes if and only if $R_2(t)$ is constant, that is, $R_2(t)=R_2(0)$ for all $t\ge 0$, with $R_2(0)\mathbf v_{\mathbf s_1}=\mathbf v_{\mathbf s_0}$. In this case, it coincides with the right-hand side of \eqref{eq:rot} if and only if $R_1(t)$ is constant and satisfies $R_1(t)=R_2^{-1}(0)$ for all $t\ge 0$. The channel $\mathcal C(t)$ is therefore given by the expression \eqref{eq:affineCaffineCA2} in the non-unital case.

Finally, we remark that eq. \eqref{eq:rot} cannot hold whenever one side of the equation is non-null. To prove this, consider the set of points defined by the right-hand side, that is 
\begin{equation}
	A = \{ (x, y) \in \mathbb{R}^2 \mid (x, y) = R_1^{-1}(t) {\bf v}_{{\bf s}_1} - {\bf v}_{{\bf s}_0}, \, t \ge 0 \}.
\end{equation}
The subset $A$ is included in the circle with equation
\begin{equation}
	(x + {\bf v}_{{\bf s}_0,1})^2 + (y + {\bf v}_{{\bf s}_0,2})^2 = \| {\bf v}_{{\bf s}_1}  \|^2.
\end{equation}
Instead, the set of points defined by the left-hand side is
\begin{equation}
	B = \{ (x, y) \in \mathbb{R}^2 \mid (x, y) = D(t) (R_2(t) {\bf v}_{{\bf s}_1}  - {\bf v}_{{\bf s}_0} ), \, t \ge 0 \},
\end{equation}
which is a subset of the ellipse with equation
\begin{equation}
	\frac{(x + e^{-a_1 t} {\bf v}_{{\bf s}_0,1})^2}{e^{-2a_1 t}} + \frac{(y + e^{-a_2 t} {\bf v}_{{\bf s}_0,2})^2}{e^{-2a_2 t}} = \|  {\bf v}_{{\bf s}_1} \|^2.
\end{equation}
Since no point remains in the intersection of $A$ and $B$ for all $t \ge 0$, eq. \eqref{eq:rot} cannot hold whenever one side is non-null. This completes the proof of the Proposition \ref{prop1.1}.
\qed 

\medskip 

\noindent The Markovian rebit channels described by eq. \eqref{eq:affineCaffineCA1} are unital and \textit{depolarizing}, i.e. dynamical maps that uniformly contract all rebit states towards the maximally mixed state at the centre of the Bloch disk, see e.g. \cite{Heinosaari:2011} for more details.

Instead, in general, Markovian rebit channels described by eq. \eqref{eq:affineCaffineCA2} are neither depolarizing nor unital. In particular, if $R_2(0)=I_2$, then the first three inequalities of the CP conditions \eqref{eq:channelconstraints} are automatically satisfied and the only non-trivial constraint that remains gives
\begin{equation}
	\label{eq:adm}
	\frac{\left(b_1/a_1\right)^2\left(1-e^{-a_1t}\right)^2}{\left(1+e^{-a_1t}\right)^2-e^{-2a_2t}}
	+\frac{\left(b_2/a_2\right)^2\left(1-e^{-a_2t}\right)^2}{\left(1+e^{-a_2t}\right)^2-e^{-2a_1t}}\le 1, \quad t\geq 0.
\end{equation}

\begin{remark}\label{rem1.1}
	Inequality \eqref{eq:adm} implies that the stationary state of ${\cal C}_{\mathbb A}(t)$ in eq. \eqref{eq:affineCaffineCA2} has a Bloch vector 
	\beq
	{\bf v}_{{\bf s}_0}=\left(\frac{b_1}{a_1},\frac{b_2}{a_2}\right)
	\eeq 
	which belongs to the Bloch disk $\mathcal D$. However, Figure \ref{figCP} shows that the converse does not hold: inequality \eqref{eq:adm} is not equivalent to the requirement that $\mathbf{v}_{\mathbf{s}_0}\in \cal D$. 
	
	It is not straightforward to extract from inequality \eqref{eq:adm} explicit conditions on $a_1$, $a_2$, $b_1$, and $b_2$ that guarantee the complete positivity of dynamical maps of the form \eqref{eq:affineCaffineCA2}. This observation motivates the analysis conducted via the GKSL equation that will be carried out in Section \ref{sec:GKSL}.

	\begin{figure}[ht!]
		\centering
		\includegraphics[scale=0.25]{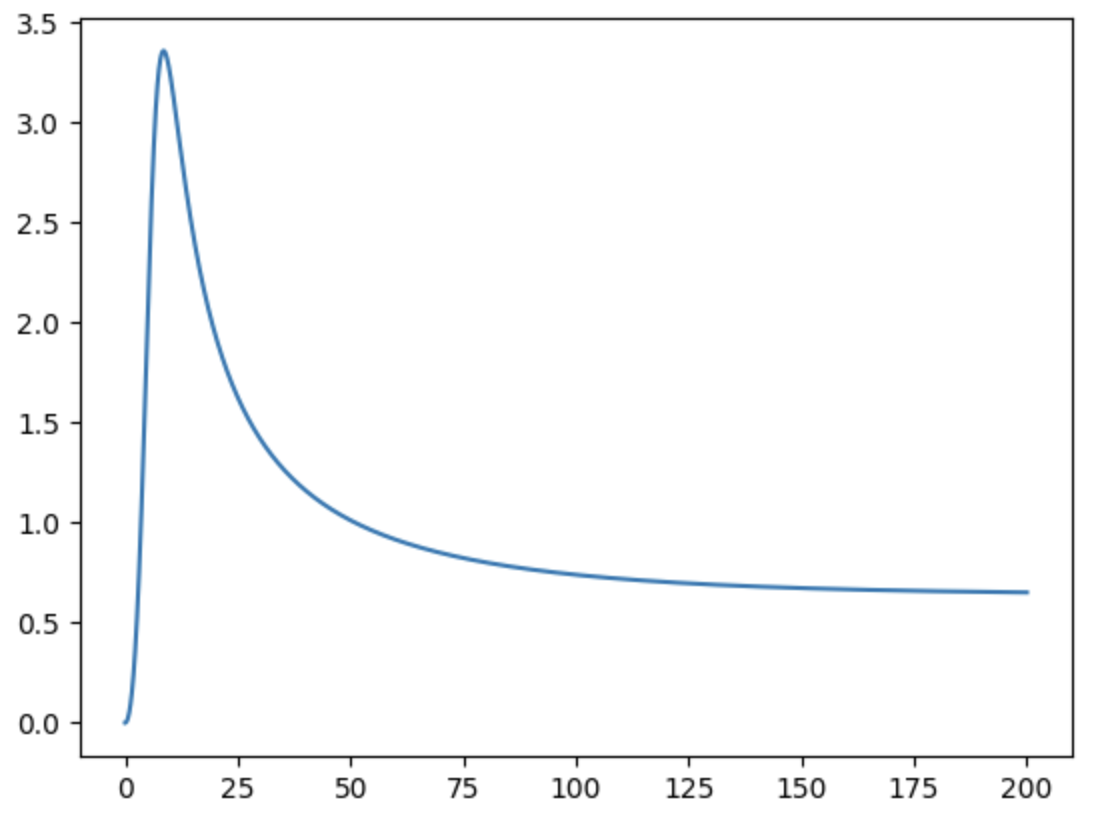}
		\caption{Curve representing the numerical value of the left hand side of inequality \eqref{eq:adm}, with $a_1=0.5$, $a_2=0.01$, $b_1=0.4$, and $b_2=0$, with $t$ varying from $0$ to $200$. Despite the fact that $(b_1/a_1, b_2/a_2)\in \cal D$, the graph shows a violation of inequality \eqref{eq:adm}.}
		\label{figCP}
	\end{figure}
	
\end{remark}

An explicit counterexample showing that the affine component of a Markovian channel need not be Markovian in general is given by
\begin{equation}
	\mathcal C(t) =
	\begin{pmatrix}
		e^{-t}(1 + t) & -e^{-t} t \\
		e^{-t} t & e^{-t}(1 - t)
	\end{pmatrix}
	= R(t)\,\mathcal C_{\mathbb A}(t)\,R(t),
\end{equation}
where $R(t)$ is the rotation of angle $\theta(t)=\arctan(t)/2$.
One readily checks that
\begin{equation}
	\lambda_1(t) = e^{-t}\big(\sqrt{1+t^2}+t\big), \qquad
	\lambda_2(t) = e^{-t}\big(\sqrt{1+t^2}-t\big).
\end{equation}
These functions coincide with those appearing in \eqref{eq:lambdaqnull} for $q=0$, $l_0=-1$, and $l_1=1$, so $\mathcal C(t)$ is indeed a Markovian rebit channel.
However, since $\lambda_i(t)\neq \exp(\dot\lambda_i(0)\,t)$ for $i=1,2$, the affine component $\mathcal C_{\mathbb A}(t)$ does not define a Markovian semigroup.

\subsection{Classification of unital Markovian rebit channels}\label{subsec:unitalgeneral}

The counterexample discussed above motivates the analysis developed in this subsection and in the next, where we do not assume that the affine component $\mathcal C_{\mathbb A}(t)$ is itself Markovian. In the present subsection we provide a complete classification of \emph{unital} Markovian rebit channels, while in \ref{sec:general} we remove the unitality assumption and treat the general case.

The following notation will simplify the subsequent computations:
\beq\label{eq:l-notations} 
l_0:=\frac{\dot\lambda_1(0)+\dot\lambda_2(0)}{2}, \quad l_1:=\frac{\dot\lambda_1(0)-\dot\lambda_2(0)}{2}, \quad  l_3:=\dot\theta_1(0)+\dot\theta_2(0), \quad q:=l_3^2-l_1^2,
\eeq
\beq
\sigma_3 := \sigma_2\sigma_1=\pmat{0 & -1 \\ 1 & 0}.
\eeq 
In what follows, without loss of generality, we assume that $\dot\lambda_2(0)\le \dot\lambda_1(0)$, i.e.
\beq\label{eq:l1assumption}
l_1 \geq 0.
\eeq
Since $\lambda_1(0)=1$ and $\lambda_1(t)\le 1$ for all $t\ge 0$ by the CP conditions \eqref{eq:channelconstraints}, one necessarily has 
\beq\label{eq:necessary}
l_1+l_0=\dot\lambda_1(0)\le 0.
\eeq 
This is coherent with assumption \eqref{eq:l1assumption} only if 
\beq\label{eq:l0assumption}
l_0\leq 0.
\eeq   

\begin{proposition}\label{theo1.1}
	Given the notations above and the unital Markovian rebit channel
	\begin{equation}
		\mathcal C_0(t) = R_1(t) D(t) R_2(t)=\exp\left(\dot {\mathcal C_0}(0)t\right),
	\end{equation}
	with $R_1$, $R_2$, and $D$ continuously differentiable, it holds that
	\beq
	\mathcal C_0(t)=R_2^{-1}(0)\widetilde{\mathcal C}_0(t)R_2(0),
	\eeq 
	where $\widetilde{\mathcal C}_0(t)$ has null initial rotation angles and can be written in one of the following ways.
	\begin{enumerate}
		\item If $q=0$, then
		\begin{equation}
			\widetilde{\mathcal C}_0(t)=e^{l_0t}\left[ \sigma_0 + l_1t(\sigma_1\pm \sigma_3)\right].
			\label{eq:SC}
		\end{equation}
		
		\item If $q<0$, then 
		\begin{equation}
			\widetilde{\mathcal C}_0(t)=e^{l_0t}\left[ \cosh(t\sqrt{-q})\,\sigma_0 + \frac{\sinh(t\sqrt{-q})}{\sqrt{-q}}(l_1\sigma_1+l_3\sigma_3)\right].
			\label{eq:SB}
		\end{equation} 
		
		\item If $q>0$, then
		\begin{equation}
			\widetilde{\mathcal C}_0(t)=e^{l_0t}\left[ \cos(t\sqrt{q})\,\sigma_0 + \frac{\sin(t\sqrt{q})}{\sqrt{q}}(l_1\sigma_1+l_3\sigma_3)\right].
			\label{eq:SA}
		\end{equation}
	\end{enumerate}
	In all three cases, the necessary condition \eqref{eq:necessary}, i.e. $l_1+l_0\leq 0$, or equivalently $\dot\lambda_1(0)\leq 0$, is also sufficient for the validity of the CP Conditions \eqref{eq:channelconstraints}.
\end{proposition}

\proof Modulo an orthogonal conjugation, the initial rotation angles can always be assumed to vanish. In fact, eq. \eqref{eq:incondCA} and the initial condition imply
\beq
D(0)=I_2, \quad \theta_1(0)=-\theta_2(0), \quad  R_1(0)=R_2^{-1}(0),
\eeq 
therefore
\beq\label{eq:deftildeC}
\widetilde{\mathcal C}_0(t):=R_2(0)\,\mathcal C_0(t)\,R_2^{-1}(0)
\eeq
is a unital Markovian rebit channel with null initial angles, $\theta_1(0)=\theta_2(0)=0$, and such that 
\beq\label{eq:CfromtildeC}
\mathcal C_0(t)=R_2^{-1}(0)\widetilde{\mathcal C}_0(t)R_2(0).
\eeq 
If we write $D(t)$ and $R_i(t)$, $i=1,2$, as follows
\begin{equation}
	D(t)=\frac{\lambda_1(t)+\lambda_2(t)}{2}\sigma_0
	+\frac{\lambda_1(t)-\lambda_2(t)}{2}\sigma_1, \quad R_i(t)=\cos\theta_i(t)\sigma_0+\sin\theta_i(t)\sigma_3,
\end{equation}
by direct computation we obtain 
\beq\label{eq:tildeC0}
\begin{split}
	\widetilde {\mathcal C}_0(t)&=
	\frac{\lambda_1(t)+\lambda_2(t)}{2}\cos\left(\theta_1(t)+\theta_2(t)\right)\sigma_0+\frac{\lambda_1(t)-\lambda_2(t)}{2}\cos\left(\theta_1(t)-\theta_2(t)\right)\sigma_1\\
	& +\frac{\lambda_1(t)-\lambda_2(t)}{2}\sin\left(\theta_1(t)-\theta_2(t)\right)\sigma_2+\frac{\lambda_1(t)+\lambda_2(t)}{2}\sin\left(\theta_1(t)+\theta_2(t)\right)\sigma_3.
\end{split}
\eeq
Assuming that $R_1$, $R_2$, and $D$ are continuously differentiable, by differentiating
$\widetilde{\mathcal C}_0(t)$ at $t=0$, using $\lambda_1(0)=\lambda_2(0)=1$ and
$\theta_1(0)=\theta_2(0)=0$, we can write the generator of $\widetilde {\mathcal C}_0(t)$ as follows
\begin{equation}
	\dot{\widetilde{\mathcal C}}_0(0)
	=
	\frac{\dot\lambda_1(0)+\dot\lambda_2(0)}{2}\sigma_0
	+\frac{\dot\lambda_1(0)-\dot\lambda_2(0)}{2}\sigma_1
	+\bigl(\dot\theta_1(0)+\dot\theta_2(0)\bigr)\sigma_3,
\end{equation}
or, taking into account the notation \eqref{eq:l-notations}, as follows
\begin{equation}\label{eq:dotC0}
	\dot{\widetilde{\mathcal C}_0}(0)=l_0\sigma_0+l_1\sigma_1+l_3\sigma_3.
\end{equation}
The absence of the Pauli matrix $\sigma_2$ in the expression of the generator imposes a strong structural constraint, which will be exploited in the sequel.

The semigroup generated by $\dot{\widetilde{\mathcal C}_0}(0)$ is
\begin{equation}\label{eq:expgenerator}
	\widetilde{\mathcal C}_0(t)=\exp\left(\dot{\widetilde{\mathcal C}_0}(0)t\right)=e^{l_0t}\exp\left[(l_1\sigma_1+l_3\sigma_3)t\right].
\end{equation}
Since the matrix $(l_1\sigma_1+l_3\sigma_3)$ belongs to the Lie algebra $\mathfrak{sl}(2,\mathbb R)$, the semigroup $\widetilde{\mathcal C}_0(t)$ decomposes into the product of a contractive homothety with ratio $\exp(l_0 t)$ (since $l_0\leq 0$) and an area-preserving linear map in $\text{SL}(2,\mathbb R)$. It can be verified that, for all $k\in\mathbb{N}$,
\beq
(l_1\sigma_1+l_3\sigma_3)^{2k} = (-q)^k \sigma_0, \quad 
(l_1\sigma_1+l_3\sigma_3)^{2k+1} = (-q)^k (l_1\sigma_1+l_3\sigma_3),
\eeq
which leads to 
\begin{equation}\label{eq:powerseries}
	\exp\left[(l_1\sigma_1+l_3\sigma_3)t\right]
	= \sum_{k=0}^\infty \frac{(-q)^k}{(2k)!} t^{2k}\,\sigma_0 
	+ \sum_{k=0}^\infty \frac{(-q)^k}{(2k+1)!} t^{2k+1}(l_1\sigma_1+l_3\sigma_3). 
\end{equation}
Following \cite{Rubilar:2020}, the matrix $(l_1\sigma_1+l_3\sigma_3)$ belongs to one and only one of the adjoint orbits of $\mathfrak{sl}(2,\mathbb R)$, which are characterized by the sign of $q=l_3^2-l_1^2$. This observation and eq. \eqref{eq:powerseries}  lead, after straightforward computations, to the following possible explicit forms of $\widetilde{\mathcal C}_0(t)$:
\beq\label{eq:caseschannel}
\widetilde{\mathcal C}_0(t)=e^{l_0t} \cdot
\begin{cases}
	\displaystyle
	\cos(t\sqrt{q})\,\sigma_0
	+
	\frac{\sin(t\sqrt{q})}{\sqrt{q}}\,
	(l_1\sigma_1+l_3\sigma_3),
	& q>0\\[1.2ex]
	\displaystyle
	\cosh(t\sqrt{-q})\,\sigma_0
	+
	\frac{\sinh(t\sqrt{-q})}{\sqrt{-q}}\,
	(l_1\sigma_1+l_3\sigma_3),
	& q<0 \\[1.2ex]
	\displaystyle
	\sigma_0
	+
	t\,l_1(\sigma_1\pm \sigma_3),
	& q=0.
\end{cases}
\eeq
Note in particular that if $l_1=l_3=0$, then $\widetilde{\mathcal C}_0(t)=e^{l_0t}\sigma_0$. This verifies the explicit expressions of $\widetilde{\mathcal C}_0(t)$ appearing in Proposition \ref{theo1.1}. 

The only claim that remains to be established concerns the complete positivity (CP) conditions.
We first observe that, for the unital channel $\widetilde{\mathcal C}_0(t)$, the CP conditions \eqref{eq:channelconstraints} are equivalent to
\beq\label{eq:CPunital}
\lambda_1(t)\leq 1 \quad \text{and} \quad \lambda_2(t)\leq 1 , \qquad \forall t\geq 0 .
\eeq
Indeed, $\widetilde{\mathcal C}_0(t)$ is invertible for all $t\geq 0$, which implies $\det D(t)\neq 0$, or equivalently $\lambda_1(t)\lambda_2(t)\neq 0$ for all $t\geq 0$. Since $\lambda_1(0)=\lambda_2(0)=1$, continuity ensures that both $\lambda_1(t)$ and $\lambda_2(t)$ remain strictly positive for all $t\geq 0$. Under this condition, the CP constraints \eqref{eq:channelconstraints} reduce precisely to the inequalities in \eqref{eq:CPunital}.

In the remainder of the proof, we derive the explicit expressions of $\lambda_1(t)$ and $\lambda_2(t)$ by equating \eqref{eq:tildeC0} and \eqref{eq:caseschannel} in each of the three cases distinguished by the sign of $q$. We then show that, in every case, the CP conditions \eqref{eq:CPunital} are equivalent to the inequality $l_1+l_0\le 0$, that is, $\dot\lambda_1(0)\le 0$.

\begin{case}[$q=0$]
	\normalfont{Suppose $q=l_3^2-l_1^2=0$, and consider the expression 
		\begin{equation}
			\widetilde{\mathcal C}_0(t)=e^{l_0t}\left(\sigma_0+l_1t(\sigma_1+\sigma_3)\right),
		\end{equation}
		corresponding to the choice of the positive sign between $\sigma_1$ and $\sigma_3$, the case of the negative sign will be discussed below. Equating this expression with \eqref{eq:tildeC0} yields the system
		\beq
		\begin{cases}
			(\lambda_1(t)+\lambda_2(t))\cos\left(\theta_1(t)+\theta_2(t)\right) = 2e^{l_0 t}\\
			(\lambda_1(t)-\lambda_2(t))\cos\left(\theta_1(t)-\theta_2(t)\right)=2l_1te^{l_0 t}\\
			(\lambda_1(t)-\lambda_2(t))\sin\left(\theta_1(t)-\theta_2(t)\right)=0\\
			(\lambda_1(t)+\lambda_2(t))\sin\left(\theta_1(t)+\theta_2(t)\right)=2l_1te^{l_0 t}.
		\end{cases}
		\eeq
		We have that $\lambda_1(t)=\lambda_2(t)$ for all $t\ge 0$ if and only if $l_1=0$ and, in this case, $\widetilde{\mathcal C}_0(t)=e^{l_0t}\sigma_0$. Instead, if $l_1\neq 0$, the third equation implies  $\theta_1(t)=\theta_2(t)=:\theta(t)$ by continuity and the initial condition on $\theta_i(t)$. Taking the ratio of the last and the first equation
		gives
		\begin{equation}
			\theta(t)=\frac 1 2 \arctan(l_1t).
		\end{equation}
		So, $\cos (2\theta(t))=1/\sqrt{1+(l_1t)^2}$ and by the first equation we obtain
		\beq\label{eq:l1plusl2}
		\lambda_1(t)+\lambda_2(t)=2e^{l_0t}\sqrt{1+(l_1t)^2}.
		\eeq
		Moreover, 
		\beq\label{eq:l1timesl2}
		\lambda_1(t)\lambda_2(t)=\det (D(t))=\det (\widetilde{\mathcal C}_0(t))=\det (e^{l_0t}\exp\left[(l_1\sigma_1+l_3\sigma_3)t\right])=e^{2l_0t},
		\eeq 
		where the last equality follows from the fact that $\sigma_1$ and $\sigma_3$ are traceless. \eqref{eq:l1plusl2} and \eqref{eq:l1timesl2} imply
		\beq\label{eq:lambdaqnull}
		\lambda_1(t)=e^{l_0t}\left(\sqrt{1+(l_1t)^2}+ l_1t\right), \quad \lambda_2(t)=e^{l_0t}\left(\sqrt{1+(l_1t)^2}- l_1t\right).
		\eeq
		Let us now show that the necessary condition \eqref{eq:necessary}, namely $\dot\lambda_1(0)\le 0$, is also sufficient for complete positivity.
		From \eqref{eq:l1assumption} we have $l_1\ge 0$, hence $\lambda_2(t)\le \lambda_1(t)$ for all $t\ge 0$, and it is therefore enough to prove that \eqref{eq:necessary} implies $\lambda_1(t)\le 1$ for all $t\ge 0$. Since $\lambda_1(0)=1$, $\dot\lambda_1(0)\le 0$, and
		$\lambda_1(t)\to 0$ as $t\to +\infty$, it is sufficient to verify that $\dot\lambda_1(t)$ does not vanish for any $t>0$. We have
		\beq
		\begin{split}
			\dot\lambda_1(t) &= e^{l_0t}\left[l_0\left(\sqrt{(l_1t)^2+1}+l_1t\right) 
			+ l_1\left(\frac{l_1t}{\sqrt{(l_1t)^2+1}}+1\right)\right] \\
			& = e^{l_0t}\left[
			l_0{\sqrt{(l_1t)^2+1}}\left(1+\frac{l_1t}{\sqrt{(l_1t)^2+1}}\right)+l_1\left(\frac{l_1t}{\sqrt{(l_1t)^2+1}}+1\right)
			\right]\\
			& =e^{l_0t}\left[\left(l_0{\sqrt{(l_1t)^2+1}}+l_1\right)\left(\frac{l_1t}{\sqrt{(l_1t)^2+1}}+1\right)
			\right].
		\end{split}
		\eeq
		The derivative $\dot\lambda_1(t)$ vanishes if and only if $l_1=-l_0{\sqrt{(l_1t)^2+1}}$, but from \eqref{eq:l0assumption} we know that $l_0\le 0$, so $\dot\lambda_1(t)=0$ if and only if $l_1\geq -l_0$, with equality only at $t=0$. However, this violates \eqref{eq:necessary} and so $\dot\lambda_1(t)\neq 0$ for all $t>0$.
		
		The case of the negative sign between $\sigma_1$ and $\sigma_3$ simply amounts to exchanging $\lambda_1(t)$ with $\lambda_2(t)$ and reversing the sign of the angle function $\theta(t)$, without affecting the validity of the argument.}
\end{case}

\begin{case}[$q<0$]
	\normalfont{Suppose now $q=l_3^2-l_1^2<0$, then
		\begin{equation}\label{eq:qneg}
			\widetilde{\mathcal C}_0(t)=e^{l_0t}\left(\cosh(t\sqrt{-q})\,\sigma_0 + \frac{\sinh(t\sqrt{-q})}{\sqrt{-q}}(l_1\sigma_1+l_3\sigma_3)\right).
		\end{equation}
		Note that if $l_3=0$, then $q=-l_1^2<0$ and $\sqrt{-q}=l_1$.
		In this case the generator reduces to $l_1\sigma_1$, and therefore \eqref{eq:expgenerator} gives
		\beq
		\widetilde{\mathcal C}_0(t)
		=
		\begin{pmatrix}
			e^{(l_0+l_1)t} & 0\\
			0 & e^{(l_0-l_1)t}
		\end{pmatrix},
		\eeq
		so, if $q<0$ and $l_3= 0$, $\lambda_1(t)=e^{(l_0+l_1)t}$ and $\lambda_2(t)=e^{(l_0-l_1)t}$ are both $\le 1$ for all $t\ge 0$ if $l_1+l_0\le 0$.
		
		If $l_3\neq 0$, equating \eqref{eq:qneg} with \eqref{eq:tildeC0} yields the system
		\beq
		\begin{cases}
			\dfrac{\lambda_1(t)+\lambda_2(t)}{2}\cos\left(\theta_1(t)+\theta_2(t)\right) = e^{l_0 t}\cosh(t\sqrt{-q})\\
			\dfrac{\lambda_1(t)-\lambda_2(t)}{2}\cos\left(\theta_1(t)-\theta_2(t)\right)=\dfrac{l_1\sinh(t\sqrt{-q})}{\sqrt{-q}}e^{l_0 t}\\
			\dfrac{\lambda_1(t)-\lambda_2(t)}{2}\sin\left(\theta_1(t)-\theta_2(t)\right)=0\\
			\dfrac{\lambda_1(t)+\lambda_2(t)}{2}\sin\left(\theta_1(t)+\theta_2(t)\right)=\dfrac{l_3\sinh(t\sqrt{-q})}{\sqrt{-q}}e^{l_0 t}.
		\end{cases}
		\eeq
		We set
		\beq\label{eq:settingfncts}
		S(t)=\frac{\lambda_1(t)+\lambda_2(t)}{2},\;
		D(t)=\frac{\lambda_1(t)-\lambda_2(t)}{2},\;
		A(t)=\theta_1(t)+\theta_2(t),\;
		B(t)=\theta_1(t)-\theta_2(t).
		\eeq
		Notice that the case $D(t)=0$, i.e. $\lambda_1(t)=\lambda_2(t)$ for all $t\ge 0$ would force the second equation of the system to give
		$l_1=0$, which contradicts $q=l_3^2-l_1^2<0$, therefore the system can be analysed under the assumption $\lambda_1(t)\neq\lambda_2(t)$. 
		We have $\theta_1(t)=\theta_2(t)=:\theta(t)$, with
		\begin{equation}\label{eq:thetaqneg}
			\tan(2\theta(t))=\frac{l_3}{\sqrt{-q}}\tanh(t\sqrt{-q}).
		\end{equation} 
		Then, the third equation implies $\sin B(t)=0$, that is $\cos B(t)=\pm 1$. The choice of the plus sign is equivalent to the condition $\lambda_1(t)\ge\lambda_2(t)$ for all $t\ge 0$. Without loss of generality, assuming this choice we get
		\beq
		D(t)=\frac{l_1\sinh(t\sqrt{-q})}{\sqrt{-q}}e^{l_0t}.
		\eeq
		Moreover, by squaring and adding the first and fourth equations, and using $-q = l_1^2 - l_3^2$ and $\cosh^2 x = 1 + \sinh^2 x$, we obtain, after taking the square root,
		\beq
		S(t)=\frac{e^{l_0t}}{\sqrt{-q}}
		\sqrt{l_1^2\sinh^2(t\sqrt{-q})-q}.
		\eeq
		Finally, by adding and subtracting $S(t)$ and $D(t)$ we get
		\beq\label{eq:lambda1qneg}
		\lambda_1(t) = \frac{e^{l_0 t}}{\sqrt{-q}}
		\left[\, \sqrt{\sinh^2\bigl(t\sqrt{-q}\bigr)l_1^2 - q}
		+ \sinh\bigl(t\sqrt{-q}\bigr)l_1 \,\right],
		\eeq
		\beq\label{eq:lambda2qneg}
		\lambda_2(t) = \frac{e^{l_0 t}}{\sqrt{-q}}
		\left[\, \sqrt{\sinh^2\bigl(t\sqrt{-q}\bigr)l_1^2 - q}
		- \sinh\bigl(t\sqrt{-q}\bigr)l_1 \,\right].
		\eeq
		We now proceed as we have done in the case $q=0$. The derivative $\dot\lambda_1(t)$ in this case vanishes if and only if the function $\dot f(t)+l_0f(t)$ vanishes, where
		\begin{gather}
			f(t) = \sqrt{\sinh^2(t\sqrt{-q})(l_1)^2 - q} + \sinh(t\sqrt{-q})\,l_1,\\
			\dot f(t) = \frac{(l_1)^2 \sqrt{-q}\cosh(t\sqrt{-q})\sinh(t\sqrt{-q})}{\sqrt{\sinh^2(t\sqrt{-q})(l_1)^2 - q}} + \sqrt{-q}\cosh(t\sqrt{-q})\,l_1.
		\end{gather}
		By setting
		\beq
		\alpha(t)=\sqrt{\sinh^2(t\sqrt{-q})(l_1)^2 - q}, \quad
		\beta(t) = \sinh(t\sqrt{-q}), \quad 
		\gamma(t) = \cosh(t\sqrt{-q}),
		\eeq
		we have
		\beq
		\begin{split}
			\dot f(t)+l_0f(t) & = \frac{l_1^2}{\alpha(t)}\sqrt{-q}\gamma(t)\beta(t)+l_1\sqrt{-q}\gamma(t) + l_0\left(\alpha(t)+l_1\beta(t)\right)\\
			& = l_1\sqrt{-q}\gamma(t)\left(1+\frac{l_1\beta(t)}{\alpha(t)}\right)+l_0\alpha(t)\left(1+\frac{l_1\beta(t)}{\alpha(t)}\right)\\
			& = \left(1+\frac{l_1\beta(t)}{\alpha(t)}\right)\left(l_1\sqrt{-q}\gamma(t)+l_0\alpha(t)\right).
		\end{split}
		\eeq 
		Since
		\begin{equation}
			1+\frac{l_1\beta(t)}{\alpha(t)}=1+\frac{l_1\sinh(t\sqrt{-q})}{\sqrt{\sinh^2(t\sqrt{-q})(l_1)^2 - q}}\neq 0,
		\end{equation}
		the derivative $\dot\lambda_1(t)$ vanishes if and only if $l_1\sqrt{-q}\gamma(t)+l_0\alpha(t)$ vanishes. Writing
		\begin{equation}
			l_1\sqrt{-q}\gamma(t)+l_0\alpha(t)=l_1\sqrt{-q}\sqrt{\sinh^2(t\sqrt{-q})+1}+l_0\sqrt{-q}\sqrt{\frac{\sinh^2(t\sqrt{-q})(l_1)^2}{-q} +1},
		\end{equation}
		if $t\geq 0$ verifies $\dot\lambda_1(t)=0$, then
		\begin{equation}
			l_1\sqrt{\sinh^2(t\sqrt{-q})+1}=-l_0\sqrt{\frac{\sinh^2(t\sqrt{-q})(l_1)^2}{-q} +1}.
		\end{equation}
		But $q=l_3^2-l_1^2<0$ implies $-l_1^2/q\geq 1$ and consequently $l_1\geq -l_0$, with equality only if $t=0$, since $l_3\neq 0$.}
\end{case}

\begin{case}[$q>0$]
	\normalfont{Finally, suppose $q=l_3^2-l_1^2>0$, then
		\begin{equation}
			\widetilde{\mathcal C}_0(t)=e^{l_0t}\left(\cos(t\sqrt{q})\,\sigma_0 + \frac{\sin(t\sqrt{q})}{\sqrt{q}}(l_1\sigma_1+l_3\sigma_3)\right).
		\end{equation}
		Equating this expression with \eqref{eq:tildeC0} yields the system
		\beq
		\begin{cases}
			\dfrac{\lambda_1(t)+\lambda_2(t)}{2}\cos\left(\theta_1(t)+\theta_2(t)\right) = \cos(t\sqrt q)e^{l_0 t}\\
			\dfrac{\lambda_1(t)-\lambda_2(t)}{2}\cos\left(\theta_1(t)-\theta_2(t)\right)=\dfrac{l_1\sin(t\sqrt q)}{\sqrt q}e^{l_0 t}\\
			\dfrac{\lambda_1(t)-\lambda_2(t)}{2}\sin\left(\theta_1(t)-\theta_2(t)\right)=0\\
			\dfrac{\lambda_1(t)+\lambda_2(t)}{2}\sin\left(\theta_1(t)+\theta_2(t)\right)=\dfrac{l_3\sin(t\sqrt q)}{\sqrt q}e^{l_0 t}.
		\end{cases}
		\eeq
		From the second and the third equation, it follows that $\lambda_1(t)=\lambda_2(t)$ for all $t\ge 0$ if and only if $l_1=0$. In this case $q=l_3^2>0$, $\sqrt{q}=|l_3|$, and from \eqref{eq:caseschannel} we get 
		\beq
		\begin{split}
			\widetilde{\mathcal C}_0(t)
			& =
			e^{l_0t}\left(\cos(|l_3|t)\,\sigma_0+\sin(|l_3|t)\,\sigma_3\right)\\
			&=\begin{pmatrix}
				e^{l_0t} & 0 \\ 0 & e^{l_0t}
			\end{pmatrix}
			R(|l_3|t) = \begin{pmatrix}
				e^{l_0t} & 0 \\ 0 & e^{l_0t}
			\end{pmatrix}
			R(|\dot\theta_1(0)+\dot\theta_2(0)|t),
		\end{split}
		\eeq
		so, if $q>0$ and $l_1= 0$, $\lambda_1(t)=\lambda_2(t)=e^{l_0t} \le 1$ for all $t\ge 0$.
		
		We assume now that $l_1\neq 0$, so $\lambda_1(t)\neq \lambda_2(t)$ and $\theta_1(t)=\theta_2(t)=:\theta(t)$, with
		\begin{equation}\label{eq:thetaqpos}
			\tan(2\theta(t))=\frac{l_3}{\sqrt{q}}\tan(t\sqrt{q}).
		\end{equation}
		Solving the system as in the case $q<0$, we obtain
		\beq\label{eq:lambdasqpositive1}
		\lambda_1(t) = \frac{e^{l_0 t}}{\sqrt{q}}\left[\sqrt{\sin^2(t\sqrt{q})(l_1)^2+q} 
		+ \sin(t\sqrt{q})\,l_1\right], 
		\eeq 
		\beq\label{eq:lambdasqpositive2}
		\lambda_2(t) = \frac{e^{l_0 t}}{\sqrt{q}}\left[\sqrt{\sin^2(t\sqrt{q})(l_1)^2+q} 
		- \sin(t\sqrt{q})\,l_1\right].
		\eeq
		Contrary to the previous situations, it is no longer true that $\lambda_1(t)\geq\lambda_2(t)$ for all $t\geq 0$, see Figure \ref{figeigen_eq} for an illustration. It follows that, in this case, we must consider both scaling functions. 
		\begin{figure}[ht]
			\centering
			\includegraphics[scale=0.25]{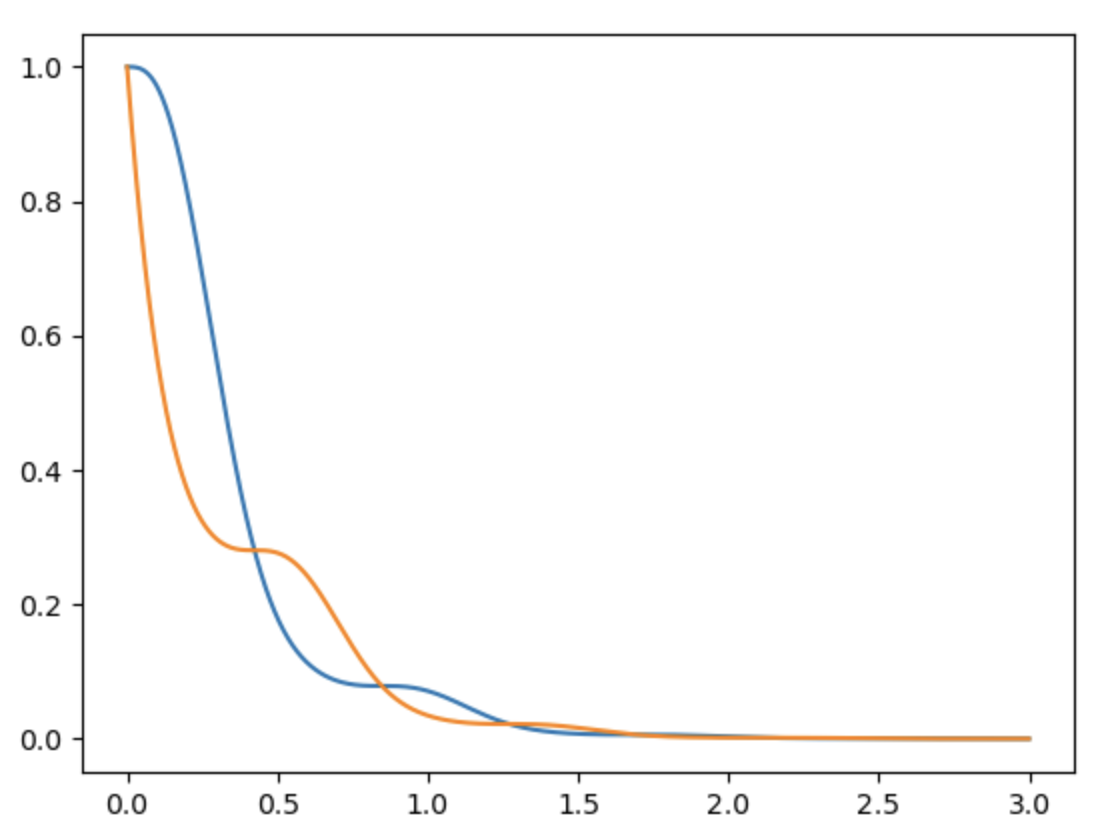}
			\caption{Graphs of $\lambda_1(t)$ and $\lambda_2(t)$ in \eqref{eq:lambdasqpositive1} and \eqref{eq:lambdasqpositive2}, with $l_0=-3$, $l_1=3$, and $l_3=8$.}
			\label{figeigen_eq}
		\end{figure} 
		
		We have that $\dot\lambda_1(t)$ vanishes if and only if
		\begin{equation}
			\left(1+\frac{l_1\beta(t)}{\alpha(t)}\right)\left(l_1\sqrt{q}\gamma(t)+l_0\alpha(t)\right)=0.
		\end{equation}
		where
		\beq
		\alpha(t)=\sqrt{\sin^2(t\sqrt{q})l_1^2+q}, \quad
		\beta(t) = \sin(t\sqrt{q}), \quad
		\gamma(t)=\cos(t\sqrt{q}).
		\eeq
		Since
		\begin{equation}
			1+\frac{l_1\beta(t)}{\alpha(t)}=1+\frac{l_1\sin(t\sqrt{q})}{\sqrt{\sin^2(t\sqrt{q})l_1^2+q}}\neq 0,
		\end{equation}
		$\dot\lambda_1(t)$ vanishes if and only if
		\begin{equation}
			l_1=-l_0\frac{1}{\cos(t\sqrt q)}\sqrt{\frac{\sin^2(t\sqrt{q})l_1^2}{q}+1}.
		\end{equation}
		If $t\ge 0$ is such that $\dot\lambda_1(t)=0$, then $\cos(t\sqrt{q})>0$ and $l_1\geq -l_0$, with equality if and only if $t\sqrt q=0$ modulo $2\pi$. But at these points 
		\begin{equation}
			\lambda_1\left(\frac{2k\pi}{\sqrt q}\right)= \exp\left(l_0\frac{2k\pi}{\sqrt q}\right)\leq 1,
		\end{equation}
		with equality if and only if $k=0$, i.e. $t=0$. In the same way $\dot\lambda_2(t)$ vanishes if and only if\begin{equation}
			\left(1-\frac{l_1\beta(t)}{\alpha(t)}\right)\left(l_0\alpha(t)-l_1\sqrt{q}\gamma(t)\right)=0,
		\end{equation}
		that is, if and only if
		\begin{equation}
			l_1=l_0\frac{1}{\cos(t\sqrt q)}\sqrt{\frac{\sin^2(t\sqrt{q})l_1^2}{q}+1}.
		\end{equation}
		If $t\ge 0$ is such that $\dot\lambda_2(t)=0$, then $\cos(t\sqrt{q})<0$ and $l_1\geq -l_0$, with equality if and only if $t\sqrt q=\pi$ modulo $2\pi$. But at these points
		\begin{equation}
			\lambda_2\left(\frac{\pi + 2k\pi}{\sqrt q}\right)= \exp\left(l_0\frac{\pi+2k\pi}{\sqrt q}\right)<1.
	\end{equation}}
\end{case}
\noindent With the analysis of the three cases $q=0,q<0,$ and $q>0$, Proposition \ref{theo1.1} is fully proven. 
\qed 

\medskip

\noindent The following corollary is an immediate consequence of the results obtained in the proof of Proposition \ref{theo1.1}.

\begin{corollary}\label{eq:corRDR}
	The unital Markovian rebit channel $\widetilde{\mathcal C}_0(t)$ can be written as
	\beq\label{eq:RDR}
	\widetilde{\mathcal C}_0(t) = R(t)D(t)R(t),
	\eeq 
	where the rotation matrix $R(t)=R(\theta(t))$ and $D(t)=\operatorname{diag}(\lambda_1(t),\lambda_2(t))$ are as follows.
	\bi
	\setlength{\itemsep}{0.8em}
	\item $q=0$ and $l_1=0:$  $R(t)=I_2$ and $D(t)=\operatorname{diag}(e^{l_0t},e^{l_0t})$.
	\item $q=0$ and $l_1\neq 0:$  $\theta(t)=\pm \arctan(l_1t)/2$ and $\lambda_{1,2}(t)=e^{l_0t}\left(\sqrt{1+(l_1t)^2}\pm l_1t\right)$.
	\item $q<0$ and $l_3=0:$ $R(t)=I_2$ and $D(t)=\operatorname{diag}\left(e^{(l_0+l_1)t},\,e^{(l_0-l_1)t}\right)$.
	\item $q<0$ and $l_3\neq 0:$ $\theta(t)$ satisfies \eqref{eq:thetaqneg} and $\lambda_{1}(t)$, $\lambda_{2}(t)$ are given by \eqref{eq:lambda1qneg} and \eqref{eq:lambda2qneg}.
	\item $q>0$ and $l_1=0:$ $\theta(t)=|l_3|t/2$ and $D(t)=\operatorname{diag}(e^{l_0t},e^{l_0t})$.
	\item $q>0$ and $l_1\neq 0:$ $\theta(t)$ satisfies \eqref{eq:thetaqpos} and $\lambda_{1}(t)$, $\lambda_{2}(t)$ are given by \eqref{eq:lambdasqpositive1} and \eqref{eq:lambdasqpositive2}.
	\ei 
\end{corollary}

\begin{remark}
	In light of Proposition \ref{theo1.1} and Corollary \ref{eq:corRDR}, the geometric interpretation of the coefficients $l_0,l_1,l_3$ is the following: $l_0$ and $l_1$ control the isotropic and anisotropic contraction of the Bloch disk, respectively, while $l_3$ accounts for rotational effects.
\end{remark}

\subsection{Classification of Markovian rebit channels: the general case}\label{sec:general}

In the previous subsection, we classified unital Markovian rebit channels, as stated in Proposition \ref{theo1.1}. Here, we extend that result to include the non-unital case, thereby obtaining the most general classification of Markovian rebit channels.

\begin{proposition}
	\label{coro1.1}
	Let $\mathcal C(t)=R_1(t)\mathcal C_{\mathbb A}(t)R_2(t)$ be a Markovian rebit channel with 
	\begin{equation}
		{\mathcal C}_{\mathbb A}(t)(\vs)=D(t)\vs+{\bf w}(t), \quad \vs\in \cal D,
	\end{equation}
	where $D(t)=\operatorname{diag}(\lambda_1(t),\lambda_2(t))$ and ${\bf w(t)}\in \cal D$ for all $t\ge 0$.
	Assume that the scaling functions $\lambda_1(t)$ and $\lambda_2(t)$ and ${\bf w}(t)$ satisfy the CP conditions in eq. \eqref{eq:channelconstraints} for all $t\geq 0$ and that $R_1$, $R_2$, and $D$ are continuously differentiable. 
	Then, there exist a unital Markovian rebit channel 
	
	\begin{equation}
		\mathcal C_0(t)=R_2^{-1}(0)R(t)D(t)R(t)R_2(0)
	\end{equation} 
	belonging to the classification appearing in Proposition \ref{theo1.1} and Corollary \ref{eq:corRDR}, and ${\bf W}(t)\in\cal D$ for all $t\ge 0$, with ${\bf W}(0)=\bf 0$, such that the action of $\mathcal C(t)$ can be written as
	\begin{equation}
		{\mathcal C}(t)(\vs)={\mathcal C}_0(t)(\vs)+{\bf W}(t), \quad \vs\in \cal D,
	\end{equation}
	with
	\begin{equation}
		R_1(t)=R_2^{-1}(0)R(t), \quad R_2(t)=R(t)R_2(0),
	\end{equation} 
	and, whenever $\dot{\mathcal C}_0(0)$ is invertible, 
	\begin{equation}
		{\bf W}(t)=R_2^{-1}(0)R(t){\bf w}(t)= \left[\dot{\mathcal C}_0(0)\right]^{-1}\left(\mathcal C_0(t)(\dot{\bf W}(0))-\dot{\bf W}(0)\right),
	\end{equation}
	otherwise, ${\bf W}(t)=\bf 0$ for all $t\ge 0$.
\end{proposition}

\proof By eq. \eqref{eq:gen}, the general action of a Markovian rebit channel is 
\begin{equation}
	{\mathcal C}(t)(\vs)={\mathcal C}_0(t)(\vs)+R_1(t){\bf w}(t).
\end{equation}
By \eqref{eq:CfromtildeC} and Corollary \ref{eq:corRDR}, the unital Markovian channel $\mathcal C_0(t)=R_1(t)D(t)R_2(t)$ is given by
\begin{equation}
	\mathcal C_0(t)=R_2^{-1}(0)\widetilde{\mathcal C}_0(t)R_2(0)=R_2^{-1}(0)R(t)D(t)R(t)R_2(0).
\end{equation}
In homogeneous coordinates, the affine action of $\mathcal C(t)$ is represented by the $3\times3$ matrix 
\begin{equation}
	C(t) = \pmat{1 & \mathbf{0}^T\\
		{\bf W}(t) & \mathcal C_0(t)} , 
\end{equation}
where ${\bf W}(t)=R_2^{-1}(0)R(t){\bf w}(t)$ is such that
\begin{equation}
	\mathcal C(t)(\vs)=\pmat{1 & \mathbf{0}^T\\
		{\bf W}(t) & \mathcal C_0(t)}\pmat{1\\ \vs}.
\end{equation}
Since ${\bf w}(0)=\mathbf{0}$, the infinitesimal generator $\dot{\mathcal C}(0)$ of the channel $\mathcal C(t)$ is then given by the matrix
\begin{equation}
	\dot{\mathcal C}(0) = \pmat{0 & \mathbf{0}^T\\
		\dot{{\bf W}}(0) & \dot{\mathcal C}_0(0)} , 
\end{equation}
with $\dot{{\bf W}}(0)=R_2^{-1}(0)R(0){\dot{\bf w}}(0)$. It can be checked the powers of $\dot{\mathcal C}(0)$ can be written as 
\begin{equation}
	\left[\dot{\mathcal C}(0)\right]^k = \pmat{0 & \mathbf{0}^T\\
		[\dot{\mathcal C}_0(0)]^{k-1}(\dot{{\bf W}}(0)) & [\dot{\mathcal C}_0(0)]^{k}} , 
\end{equation}
thus the exponential series gives
\begin{equation}
	{\bf W}(t)=\left[\dot{\mathcal C}_0(0)\right]^{-1}\left(\mathcal C_0(t)-I_2\right)\dot{\bf W}(0),
\end{equation}
provided that $\dot{\mathcal C}_0(0)$ is invertible. By \eqref{eq:dotC0}, we have 
\beq
\det\left(\dot{\mathcal C}_0(0)\right)=\det\left(\dot{\widetilde{\mathcal C}_0}(0)\right)=0
\eeq 
if and only if $\dot\lambda_1(0)\dot\lambda_2(0)+l_3^2=0$. If the CP conditions hold, then $\dot\lambda_1(0)\dot\lambda_2(0)\ge 0$ and so the determinant is null if and only if  $\dot\lambda_1(0)\dot\lambda_2(0)=l_3=0$. But then that at least one of $\lambda_1(t)$ and $\lambda_2(t)$ is 1 for all $t\ge 0$, so the CP conditions imply ${\bf W}(t)=\bf 0$ for all $t\ge 0$.
\qed 

\begin{remark} Markovian channels whose affine components are themselves Markovian, classified in Proposition \ref{prop1.1}, arise as particular cases of the previous result when $q=0$ and $l_1=0$, $q<0$ and $l_3=0$, or $q>0$ and $l_1=0$. 
\end{remark}

\section{Markovian rebit channels and GKSL master equation}\label{sec:GKSL}

As already noted in Remark \ref{rem1.1}, deriving explicit constraints on ${\bf w}(t)$ and on the scaling functions $\lambda_i(t)$, $i=1,2$, that guarantee complete positivity is not straightforward, even in the simplified setting where the affine components are Markovian.

Here we show that it is possible to obtain a simple characterisation of complete positivity entirely determined by the parameters $\dot\lambda_i(0)$ and $\dot w_i(0)$, $i=1,2$. To this end, we must establish a connection between the analysis of the previous section and the GKSL master equation.

We start quoting the following theorem, see e.g.  \cite{Breuer:2002, Gorini:76, Lindblad:76}, and then we adapt it to rebit systems.

\begin{theorem}\label{th:GKSL} 
	Let $\mathcal H$ be a complex Hilbert space of dimension $d\geq 2$. We denote $\mathcal T(\mathcal H)$ and $\mathcal B(\mathcal H)$ the sets of trace-class and bounded operators on $\mathcal H$. Let $\mathcal L:\mathcal T(\mathcal H)\longrightarrow \mathcal T(\mathcal H)$ be a linear operator. Then, $\mathcal L$ is the infinitesimal generator of a completely positive dynamical semigroup on $\mathcal T(\mathcal H)$ if and only if there exist:
	\bi 
	\item $H\in \mathcal B(\mathcal H)$ {\it with $H^\dag=H$ and Tr$(H)=0$}
	\item $F_k\in \mathcal B(\mathcal H)$, {\it with Tr$(F_k)=0$ and Tr$(F_kF_\ell^\dag)=\delta_{k\ell}$ for all $k,\ell=1,\dots,d^2-1$}
	\item $\mathcal M=\left(\alpha_{k\ell}\right)_{k,\ell=1}^{d^2-1}${\it, a self-adjoint positive semi-definite matrix,}
	\ei 
	\it such that for all $X$ in $\mathcal T(\mathcal H)$, we have:
	\begin{equation}\label{eq:Linddef}
		\mathcal L(X)=-i\left[H,X\right]+\sum_{k,\ell=1}^{d^2-1}\alpha_{k,\ell}\left(2F_kXF_\ell^\dag-\{F_\ell^\dag F_k,X\} \right).
	\end{equation}
\end{theorem}
$\mathcal{L}$ is referred to as the \textit{Lindbladian}, $M$ and $F_k$ are called \textit{Lindblad matrix} and \textit{operators}, respectively, and $\{A,B\}:=AB+BA$ is the anti-commutator. 

The differential equation
\beq
\dot X(t)=\mathcal L(X(t)), \quad t\ge 0,
\eeq 
is the GSKL master equation. 

If $X$ is a density matrix $\rho$ of a quantum system with Hilbert space $\cal H$ and Hamiltonian $H$, and if all the coefficients $\alpha_{k\ell}$ vanish, then the GKSL master equation reduces to the Liouville-von Neumann equation: $\dot{\rho}(t) = -i[H, \rho(t)]$. For this reason, the first term of the Lindbladian in \eqref{eq:Linddef} is associated with \textit{unitary evolution}, while the second term, called \textit{dissipator}, takes into account the interaction between the system and its environment.

For a single rebit: $\mathcal H=\R^2$ and $X(t)=\rho_{\bf s}(t)=(\sigma_0+s_1(t) \sigma_1 + s_2(t) \sigma_2)/2$, so the presence of the imaginary unit in \eqref{eq:Linddef} lacks a fundamental justification and must be replaced in order to obtain a Linbladian term responsible for the \textit{orthogonal evolution} of the rebit. 

Since the only meaningful orthogonal dynamics of $\rho_{\bf s}$ is given by a rotation of the Bloch vector $\vs(t)=(s_1(t),s_2(t))$, we replace the term $-i[H,X]$ in \eqref{eq:Linddef} with $[\omega\sigma_3,\rho_{\bf s}]$, with $\omega\in\R$. 

Indeed, if one retains only this contribution, the GKSL equation reduces to
\begin{equation}
	\dot\rho_{\bf s}(t)=\frac{1}{2}(\dot s_1(t)\sigma_1+\dot s_2(t)\sigma_2)=[\omega\sigma_3,\rho_{\bf s}(t)]=\omega(-s_2(t)\sigma_1+s_1(t)\sigma_2),
\end{equation}
which yields the differential system 
\begin{equation}
	\begin{cases}
		\dot s_1(t) = -\omega s_2(t)\\
		\dot s_2(t) = \omega s_1(t),
	\end{cases}
\end{equation} 
solved by
\begin{equation}
	\begin{cases}
		s_1(t) = s_1(0)\cos(\omega t) -s_2(0)\sin(\omega t) \\
		s_2(t) = s_2(0)\cos(\omega t) + s_1(0)\sin(\omega t),
	\end{cases}
\end{equation} 
thereby showing that the bracket $[\omega\sigma_3,\rho_{\bf s}]$ generates an orthogonal evolution. 

A further, independent justification for this choice follows from the observation that $\rho_{\bf s}$ belongs to the Jordan algebra $\mathcal H(2,\R)$ of symmetric $2\times 2$ real matrices, endowed with the Jordan product defined by half the anti-commutator, see e.g. \cite{Baez:12} for further information. 
As proven in \cite{Yao:2019}, the only derivations, namely Leibniz-type endomorphisms, of $\mathcal H(2,\R)$ act precisely through the bracket $[\omega\sigma_3,\rho_{\bf s}]$.

Let us now analyse the dissipator of the Lindbladian for a rebit. A natural choice for the Lindblad operators is the following family of traceless matrices:  
\beq
F_k = \frac{1}{\sqrt{2}}\sigma_k, \quad k=1,2,3,
\eeq  
in fact, $\text{Tr}(F_kF_\ell^T)=\delta_{k\ell}$. With direct computations, we obtain that the diagonal terms of $\cal L$, i.e. 
\begin{equation}
	{\mathcal L}_k(\rho_{\bf s}) := 2F_k \rho_{\bf s} F_k^T-\{F_k^T F_k,\rho_{\bf s}\}, \quad k=1,2,3,
\end{equation}
are given by
\begin{equation}\label{eq:Ldiag}
	\mathcal L_1(\rho_{\bf s}) = -s_2\sigma_2, \quad
	\mathcal L_2(\rho_{\bf s}) = -s_1\sigma_1, \quad
	\mathcal L_3(\rho_{\bf s}) = -(s_1\sigma_1+s_2\sigma_2).
\end{equation}
Similarly, the off-diagonal terms of $\cal L$
\beq
\widetilde{\mathcal L}_{k,\ell}(\rho_{\bf s}) := {\mathcal L}_{k,\ell}(\rho_{\bf s})+{\mathcal L}_{\ell,k}(\rho_{\bf s}),
\eeq
with
\beq
{\mathcal L}_{k,\ell}(\rho_{\bf s})=2F_k \rho_{\bf s} F_\ell^T-\{F_k^T F_\ell,\rho_{\bf s}\}, \quad 
{\mathcal L}_{\ell,k}(\rho_{\bf s})=2F_\ell \rho_{\bf s} F_k^T-\{F_\ell^T F_k,\rho_{\bf s}\},
\eeq
$k,\ell=1,2,3$, $k\neq \ell$, are given by  
\beq\label{eq:Loffdiag}
\widetilde{\mathcal L}_{1,2}(\rho_{\bf s})  = s_2\sigma_1+s_1\sigma_2, \quad
\widetilde{\mathcal L}_{1,3}(\rho_{\bf s})= 2\sigma_2, \quad 
\widetilde{\mathcal L}_{2,3}(\rho_{\bf s})= -2\sigma_1.
\eeq

\begin{proposition}
	The CP conditions \eqref{eq:channelconstraints} are satisfied if and only if there
	exist  $\alpha_{i,j}\in \R$, $i,j=1,2,3$, satisfying
	\begin{equation}
		\begin{cases}
			\dot \lambda_1(0) = -2(\alpha_{2,2} + \alpha_{3,3}) \\
			\dot \lambda_2(0) = -2(\alpha_{1,1} + \alpha_{3,3}) \\
			\dot w_1(0)= - 4\alpha_{2,3} \\
			\dot w_2(0)= 4\alpha_{1,3},
		\end{cases}
	\end{equation}
	and such that the matrix
	\begin{equation}
		\mathcal M=\pmat{\alpha_{1,1} & 0 &\alpha_{1,3}\\
			0 & \alpha_{2,2} & \alpha_{2,3}\\
			\alpha_{1,3} & \alpha_{2,3} & \alpha_{3,3}}
	\end{equation}
	is positive semi-definite.
\end{proposition}

\proof The key to proving the proposition lies in comparing the dynamics of a rebit density matrix $\rho_\s$ under the action of a Markovian channel $\mathcal C(t)$ with its evolution as a solution of the GKSL equation with Lindbladian $\mathcal L$.

Since $\mathcal C(t)$ is completely positive if and only if $\widetilde{\mathcal C}(t)=R_2(0)\mathcal C(t)R_2^{-1}(0)$ is, it is enough to determine the complete positivity conditions for the semigroup
\begin{equation}
	\widetilde{\mathcal C}(t)(\vs)={\widetilde{\mathcal C}}_0(t)\vs+\widetilde{\bf w}(t),
\end{equation}
with $\widetilde{\bf w}(t)=R(t){\bf w}(t).$ We have
\begin{equation}
	\dot{\widetilde{\mathcal C}}(0)(\vs)=\dot{\widetilde{\mathcal C}}_0(0)\vs+R(0)\dot{\bf w}(0),
\end{equation}
which, using eq. \eqref{eq:dotC0}, can be rewritten as
\begin{equation}
	\dot{\widetilde{\mathcal C}}(0)(\vs)=\pmat{l_0+l_1 & -l_3\\
		l_3 & l_0-l_1}\vs+R(0)\dot{\bf w}(0).
\end{equation}
By Corollary \ref{eq:corRDR}, $\theta_1(t)=\theta_2(t)=\theta(t)$, $\dot \theta(0)=\omega$ and $\theta(0)=0$, so $l_3=2\dot\theta(0)$ and
\begin{equation}\label{eq:diffC0}
	\dot{\widetilde{\mathcal C}}(0)(\vs)=\pmat{\dot\lambda_1(0) & -2\dot\theta(0)\\
		2\dot\theta(0) & \dot\lambda_2(0)}\vs+\dot{\bf w}(0).
\end{equation}
If $\rho_\s(t)$ evolves under the action of the channel $\widetilde{\mathcal C}(t)$, then 
\begin{equation}
	\rho_{\bf s}(t)=\frac 1 2(\sigma_0+s_1(t) \sigma_1 + s_2(t) \sigma_2)=\widetilde{\mathcal C}(t)\rho_{\bf s}(0)=e^{\dot{\widetilde{\mathcal C}}(0)t}\rho_{\bf s},
\end{equation}
which gives
\begin{equation}
	\dot{\rho}_{\bf s}(t)=\frac{\dot s_1(t) \sigma_1 + \dot s_2(t) \sigma_2}{2}=\dot{\widetilde{\mathcal C}}(0)\rho_{\bf s}(t)=\dot{\widetilde{\mathcal C}}(0)\frac{(\sigma_0+s_1(t) \sigma_1 + s_2(t) \sigma_2)}{2},
\end{equation}
that, thanks to eq. \eqref{eq:diffC0}, can be written as the following differential system
\begin{equation}
	\begin{cases}
		\dot s_1(t) = \dot\lambda_1(0)s_1(t) -2\dot\theta(0)s_2(t) +\dot w_1(0) \\
		\dot s_2(t) = 2\dot\theta(0)s_1(t)+\dot\lambda_2(0)s_2(t)+\dot w_2(0),
	\end{cases}
\end{equation} 
with $s_1(0)=s_1$, and $s_2(0)=s_2$. 

Now we describe the evolution of $\rho_{\bf s}$ via the GKSL master equation 
\begin{equation}
	\dot{\rho}_{\bf s}(t)=\frac{\dot s_1(t) \sigma_1 + \dot s_2(t) \sigma_2}{2}=\mathcal{L}(\rho_{\bf s}(t)),
\end{equation}
with Lindbladian
\beq\label{Lind}
\begin{split}
	\mathcal{L}(\rho_{\bf s})& =[\dot\theta(0)\sigma_3,\rho_{\bf s}]+
	\sum_{k=1}^{3}\alpha_{k,k}\mathcal{L}_k(\rho_{\bf s}) + \alpha_{1,3}\widetilde{\mathcal{L}}_{1,3}(\rho_{\bf s}) + \alpha_{2,3}\widetilde{\mathcal{L}}_{2,3}(\rho_{\bf s})\\
	&=-\left((\alpha_{2,2} + \alpha_{3,3})s_1(t) + 2\alpha_{2,3}+\dot\theta(0)s_2(t)\right)\sigma_1 \\
	& \quad -\left((\alpha_{1,1} + \alpha_{3,3})s_2(t) - 2\alpha_{1,3}-\dot\theta(0)s_1(t)\right)\sigma_2.
\end{split}
\eeq
Using eqs.  \eqref{eq:Ldiag} and  \eqref{eq:Loffdiag}, we get the differential system
\begin{equation}
	\begin{cases}
		\dot s_1 = -2(\alpha_{2,2} + \alpha_{3,3})s_1(t) - 4\alpha_{2,3} -2\dot\theta(0)s_2(t) \\
		\dot s_2 = -2(\alpha_{1,1} + \alpha_{3,3})s_2(t) + 4\alpha_{1,3} + 2\dot\theta(0)s_1(t),
	\end{cases}
\end{equation}
with $s_1(0)=s_1$, and $s_2(0)=s_2$. Comparing the two differential systems that we have obtained, we have that $\widetilde{\mathcal C}(t)$ is a Markovian rebit channel if and only if the statement of the theorem holds.  \qed

\section{Markovian rebit channels and colour perception}\label{sec:applications}

Besides its intrinsic theoretical interest, the analysis of Markovian rebit channels presented in the previous sections has a relevance in the study of the effect of illuminant changes on human colour perception, which is the main motivation behind our theoretical analysis.

The mathematical model presented in \cite{BerthierProvenzi:2022PRS,Berthier:22SIAM}, describes perceived colours as the result of quantum measurements performed by quantum effects on rebit states. If $\rho_{\bf s}$ is a density matrix, then we can write
\begin{equation}\label{eq:oppo}
	\rho_{\bf s}=\frac{1}{2}\left[ \sigma_0+ r\cos\vartheta \left(\rho_\s{(1,0)}-\rho_\s{(-1,0)}\right)+r\sin\vartheta \left(\rho_\s{(0,1)}-\rho_\s{(0,-1)}\right) \right],
\end{equation}
with $\vs=(s_1,s_2)=(r\cos\vartheta, r\sin\vartheta)$, $r\in [0,1]$, $\theta \in [0,2\pi)$. 

The state $\rho_0:=\sigma_0/2$ is the state of maximal von Neumann entropy which is natural to assume as the \textit{achromatic state}. Moreover, since orthogonal pure states are represented by antipodal points on the boundary of the Bloch disk, $(\rho_\s{(1,0)},\rho_\s{(-1,0)}$) and $(\rho_\s{(0,1)},\rho_\s{(0,-1)})$ are two pairs of mutually exclusive states. 

These considerations show that the representation of $\rho_{\bf s}$ by eq. \eqref{eq:oppo} is the exact quantum analogue of Hering's representation of colour \cite{Berthier:2020}: the {\em chromatic state} $\rho_{\bf s}$ decomposes as a sum of an {\em achromatic state} $\rho_0$ with {\em two chromatic opponency axes} $\rho_\s{(1,0)}$ vs. $\rho_\s{(-1,0)}$ and $\rho_\s{(0,1)}$ vs. $\rho_\s{(0,-1)}$, with respective {\em degrees of opponency} given by $r\cos\vartheta$ and $r\sin\vartheta$. 

The pure chromatic states $\rho_\s{(1,0)}$, $\rho_\s{(-1,0)}$, $\rho_\s{(0,1)}$, and $\rho_\s{(0,-1)}$ are {\em the four unique hues}, usually denominated as red, green, yellow and blue. 

Since $\sigma_1=\rho_\s{(1,0)}-\rho_\s{(-1,0)}$ and $\sigma_2=\rho_\s{(0,1)}-\rho_\s{(0,-1)}$, the chromatic opponency is encoded by the two Pauli matrices $\sigma_1$ and $\sigma_2$.

A quantum effect of the rebit is represented by a matrix
\begin{equation}
	\eta_{\bf e}=\pmat{e_0+e_1 & e_2\\ e_2 & e_0-e_1}=2e_0\rho{(e_1/e_0,e_2/e_0)},
\end{equation}
with $0< e_0\leq 1$, and such that the effect vector $\ve=(e_1/e_0,e_2/e_0)$ belongs to the Bloch disk. Given an effect $\eta_{\bf e}$ and a density matrix $\rho_{\bf s}$, one can perform a L\"uders measurement, see e.g. \cite{Kraus:83,Busch:97}, that gives rise to the following post-measurement generalised state 
\begin{equation}
	\label{eq:Lud}
	\psi_{\bf e}(\s):=\eta_{\bf e}^{1/2}\rho_{\bf s}\eta_{\bf e}^{1/2}=\langle{\bf e}\rangle_{\bf s} \frac{\psi_{\bf e}(\s)}{\langle{\bf e}\rangle_{\bf s}}=:\langle{\bf e}\rangle_{\bf s} \varphi_{\bf e}(\s),
\end{equation}
where $\varphi_{\bf e}(\s)$ is the post-measurement chromatic state, and 
\begin{equation}
	0\leq\langle{\bf e}\rangle_{\bf s}=e_0(1+{\ve}\cdot{\vs})\leq 1 
\end{equation}
is the expectation value of the measurement.

In the colour model, chromatic states are associated with the preparation of a visual scene and effects represent observers. Thus, {\em the generalized state $\psi_{\bf e}(\s)$ is the colour perceived by the observer ${\bf e}$ from a visual scene prepared in the chromatic state $\s$}.  It is shown in \cite{BerthierProvenzi:2022PRS} that  
\begin{equation}
	{\bf v}_{\varphi_{\ee}(\s)}=\ve \gplus \vs,
\end{equation}
where ${\ve}{\gplus}{\vs}$ is the relativistic sum of the two vectors: if $\ve=\bf 0$ then $\ve\gplus \vs := \vs$, otherwise  
\begin{equation}
	{\ve}\gplus{\vs}:=\frac{1}{1+{\ve}\cdot{\vs}}\left[{\ve}+\frac{{\vs}}{\gamma_{\ee}}+\frac{\gamma_{\bf e}({\ve}\cdot{\vs}){\ve}}{1+\gamma_{\bf e}}\right],
\end{equation}
with
\begin{equation}
	\gamma_{\bf e}=\frac{1}{\sqrt{1-\Vert{\ve}\Vert^2}}.
\end{equation}
This new description of perceived colours makes it possible to give well-founded mathematical definitions of their perceptual attributes: brightness, lightness, saturation, colourfulness, chroma, and hue, see \cite{Berthier:22SIAM} for details. It also enables the efficient implementation of white-balance algorithms by means of Lorentz boosts, see \cite{Berthier:24SPM}.

A crucial concept in quantum information is the \textit{relative entropy} between two quantum states ${\bf s}$ and ${\bf t}$. It is defined as 
\beq 
R(\rho_{{\bf s}}||\rho_{{\bf t}}):={\rm Tr}\left( \rho_{{\bf s}}\log_2\rho_{{\bf s}}-\rho_{{\bf s}}\log_2\rho_{{\bf t}}\right),
\eeq 
and it represents an important measure of distinguishability between states. In the specific case of a rebit, the relative entropy has an analytical expression given by \cite{Cortese:02}
\beq
R(\rho_{{\bf s}}||\rho_{{\bf t}})=\frac{1} 2 \log_2(1-r_{\bf s}^2)+\frac{r_{\bf s}} 2 \log_2\left(\frac{1+r_{\bf s}}{1-r_{\bf s}}\right)-\frac1 2 \log_2(1-r_{\bf t}^2) 
-\frac{r_{\bf s}\cos \theta_{{\bf s},{\bf t}}} 2 \log_2\left(\frac{1+r_{\bf t}}{1-r_{\bf t}}\right).
\eeq 
The interplay between quantum relative entropy and quantum channels is encoded by the well-known Data Processing Inequality (DPI), see e.g. \cite{Wilde:2017,Witten:2020},
\beq\label{eq:DPI}
R(\mathcal C(\rho_{\bf s})\,\|\,\mathcal C(\rho_{\bf t}))\le R(\rho_{\bf s}\,\|\,\rho_{\bf t}), \quad \forall \, \s,{\bf t},
\eeq
which implies that channels tend to decrease the distinguishability between quantum states.

In the following subsections, we analyse the role of Markovian rebit channels within the colour perception model outlined above.

\subsection{Illuminant channels and chromatic distortion}

Let us consider a visual scene prepared in such a way that an observer $\ee$ is exposed to the chromatic state $\s$ of a colour stimulus in a room illuminated by an \textit{achromatic} illuminant $\imath_a$. Physically, $\imath_a$ corresponds to a light source with a broadband spectrum. Mathematically, $\imath_a$ is described by a null Bloch vector ${\bf v}_{\imath_a}=\mathbf{0}$.

Empirically, after a sufficiently long exposure, the observer $\ee$ becomes adapted to $\imath_a$, in the sense that the observer's effect vector vanishes. Denoting by $\ee_a$ an observer adapted to an achromatic illuminant, one has ${\bf v}_{\ee_a}=\mathbf{0}$ and therefore ${\bf v}_{\ee_a}\gplus \vs=\vs$. Consequently, $\ee_a$ perceives chromatic states as they intrinsically are, namely $\varphi_{\ee_a}(\s)=\s$.

Suppose now that the scene preparation is altered by replacing $\imath_a$ with a \textit{non-neutral} illuminant $\imath$, with Bloch vector ${\bf v}_\imath\neq \mathbf{0}$. Empirically, the observer's colour perception is reported to undergo a sudden distortion for two reasons: first, chromatic states are no longer perceived in the same manner as before; second, the ability to discriminate between colours is reduced. 

This condition persists for a finite amount of time, until the chromatic adaptation mechanisms of the human visual system intervene to counteract the distortion, progressively steering perception toward a new adapted regime.

Thus, chromatic distortion can be understood as a perceptual phenomenon arising from a change in the overall configuration of the scene (the original chromatic state of the stimulus perceived under a novel illuminant), to which the observer is not yet adapted. Accordingly, at the mathematical level, it is natural to model chromatic distortion by means of a dynamical map acting \textit{on states}, and chromatic adaptation by one acting \textit{on effects}.

Here, we focus on modelling chromatic distortion by means of Markovian rebit channels. The use of quantum channels is motivated by the fact that they capture the two key features highlighted above: they are dynamical maps acting on rebit states and, by virtue of the DPI \eqref{eq:DPI}, they necessarily reduce state distinguishability. The restriction to Markovian dynamics is further justified by the validity, in the present context, of the Born assumptions underlying the GKSL equation, namely the weak coupling between the system and its environment \cite{Auletta:09}. Indeed, the illuminant influences colour perception by the observer, whereas the observer's perceptual state does not affect the illuminant.

In what follows, the four unique hues are assumed to be fixed, thereby defining a reference frame for colour perception. Accordingly, we restrict our attention to Markovian rebit channels $\mathcal C(t)$ of the form \eqref{eq:gen} with $R_1(t)=R_2(t)=I_2$ for all $t\geq 0$. More general scenarios could in principle be envisaged, however their analysis from the viewpoint of colour perception becomes substantially more involved.

We denote by $\s(t)$ the evolution of an initial chromatic state $\s$ under the action of $\mathcal C(t)$, with $\s(0)=\s$. Accordingly,
\beq\label{eq:vst}
\vs(t):=\mathcal C(t)(\vs), \quad t\ge 0,
\eeq
is the Bloch vector of the density matrix $\rho_{{\s}(t)}$, with $\vs(0)=\vs$ and $\rho_{{\s}(0)}=\rho_\s$.

Before chromatic adaptation takes place, the observer remains adapted to the achromatic illuminant, so that $\ee=\ee_a$ and
\beq
{\bf v}_{\ee_a} \gplus \vs(t)=\vs(t), \qquad t\ge 0.
\eeq
Consequently, the corresponding evolution of the post-measurement chromatic state is given by
\beq\label{eq:rhost}
\varphi^{t}_{\ee_a}(\s)=\rho_{\s(t)}, \qquad t\ge 0.
\eeq
A distinctive feature of chromatic distortion induced by an illuminant change is that, during the distortion phase, achromatic colour stimuli are no longer perceived as such, but instead \textit{acquire a chromatic component resembling that of the illuminant}. This entails that the channel $\mathcal C(t)$ must induce a displacement of the centre of the Bloch disk, and therefore cannot be unital.

From eqs. \eqref{eq:Ldiag} and \eqref{eq:Loffdiag} one sees that the diagonal dissipative terms of $\cal L$, i.e. 
$\mathcal L_k$, $k=1,2,3$, are all proportional to the Bloch components $s_1$, $s_2$, and therefore vanish when evaluated on the achromatic state. By contrast, the mixed dissipative terms $\widetilde{\mathcal L}_{1,3}$, $\widetilde{\mathcal L}_{2,3}(\rho_{\bf s})$
are independent of the Bloch vector $\vs$. As a consequence, non-zero coefficients $\alpha_{1,3}$,  $\alpha_{2,3}$ produce a constant contribution to the Lindbladian that  generates a shift of the centre of the Bloch disk. 

Hence, in order for the channel $\mathcal C(t)$ to displace the centre of the Bloch disk, the coefficients $\alpha_{1,3}$ and $\alpha_{2,3}$ cannot both vanish. In the purely dissipative setting considered here, these coefficients can be interpreted as coupling parameters linking the illuminant component encoded by $\sigma_3$ to the chromatic opponent degrees of freedom associated with $\sigma_1$ and $\sigma_2$. Concretely, they quantify how the presence of a non-neutral illuminant induces systematic drifts along the chromatic opponent directions. This effect is not a rotation of the opponency axes, but a direct dissipative coupling that converts the illuminant bias into a chromatic displacement, affecting even the achromatic state.

The above considerations motivate the following definition.

\begin{definition}
	\label{def3.1} An illuminant channel is a Markovian rebit channel generated by the Lindbladian
	\beq\label{Lindill}
	\begin{split}
		\mathcal{L}(\rho_{\bf s}) &=
		\alpha_{1,1}\mathcal{L}_1(\rho_{\bf s}) + \alpha_{2,2}\mathcal{L}_2(\rho_{\bf s}) + \alpha_{3,3}\mathcal{L}_3(\rho_{\bf s}) + \alpha_{1,3}\widetilde{\mathcal{L}}_{1,3}(\rho_{\bf s}) + \alpha_{2,3}\widetilde{\mathcal{L}}_{2,3}(\rho_{\bf s})\\
		& =-\left((\alpha_{2,2} + \alpha_{3,3})s_1 + 2\alpha_{2,3}+\dot\theta(0)s_2\right)\sigma_1\\
		& \quad - \left((\alpha_{1,1} + \alpha_{3,3})s_2 - 2\alpha_{1,3}-\dot\theta(0)s_1\right)\sigma_2.
	\end{split}
	\eeq
\end{definition}

\noindent The GKSL equation corresponding to $\cal L$ leads to the following differential system 
\begin{equation}\label{eq:diffL}
	\begin{cases}
		\dot s_1 = -2(\alpha_{2,2} + \alpha_{3,3})s_1(t) - 4\alpha_{2,3} \\
		\dot s_2 = -2(\alpha_{1,1} + \alpha_{3,3})s_2(t) + 4\alpha_{1,3},
	\end{cases}
\end{equation}
solved by 
\begin{equation}\label{eq:system}
	\begin{cases}
		s_1(t) = s_1(0)e^{-2(\alpha_{2,2}+\alpha_{3,3})t} - \dfrac{2\alpha_{2,3}}{\alpha_{2,2}+\alpha_{3,3}}\left(1 - e^{-2(\alpha_{2,2}+\alpha_{3,3})t}\right) \\[3mm]
		s_2(t) = s_2(0)e^{-2(\alpha_{1,1}+\alpha_{3,3})t} + \dfrac{2\alpha_{1,3}}{\alpha_{1,1}+\alpha_{3,3}}\left(1 - e^{-2(\alpha_{1,1}+\alpha_{3,3})t}\right).
	\end{cases}
\end{equation}
The stationary state $\s_0$ has Bloch vector
\begin{equation}
	{\bf v}_{\s_0}=(s_{0,1},s_{0,2}) = \left(-\frac{2\alpha_{2,3}}{\alpha_{2,2} + \alpha_{3,3}},\frac{2\alpha_{1,3}}{\alpha_{1,1} + \alpha_{3,3}} \right).
\end{equation}
The state ${\bf s}_0$ characterizes the asymptotic behaviour of the purely dissipative dynamics. In the absence of chromatic adaptation, the initial chromatic components $s_i(0)$ are exponentially suppressed and progressively replaced by the fixed offsets determined by the Lindblad coefficients. 


By definition, an illuminant channel is a Markovian rebit channel whose affine component is Markovian, see the system \eqref{eq:system}.  Moreover, the differential system \eqref{eq:diffL} shows that an illuminant channel is unital if and only if $\alpha_{1,3}=\alpha_{2,3}=0$, which is equivalent to the fact that the illuminant $\imath$ is achromatic.

Figure \ref{dynLind} illustrates the dynamics generated by two illuminant channels.

\begin{figure}[htbp] 
	\center
	\includegraphics[scale=0.18]{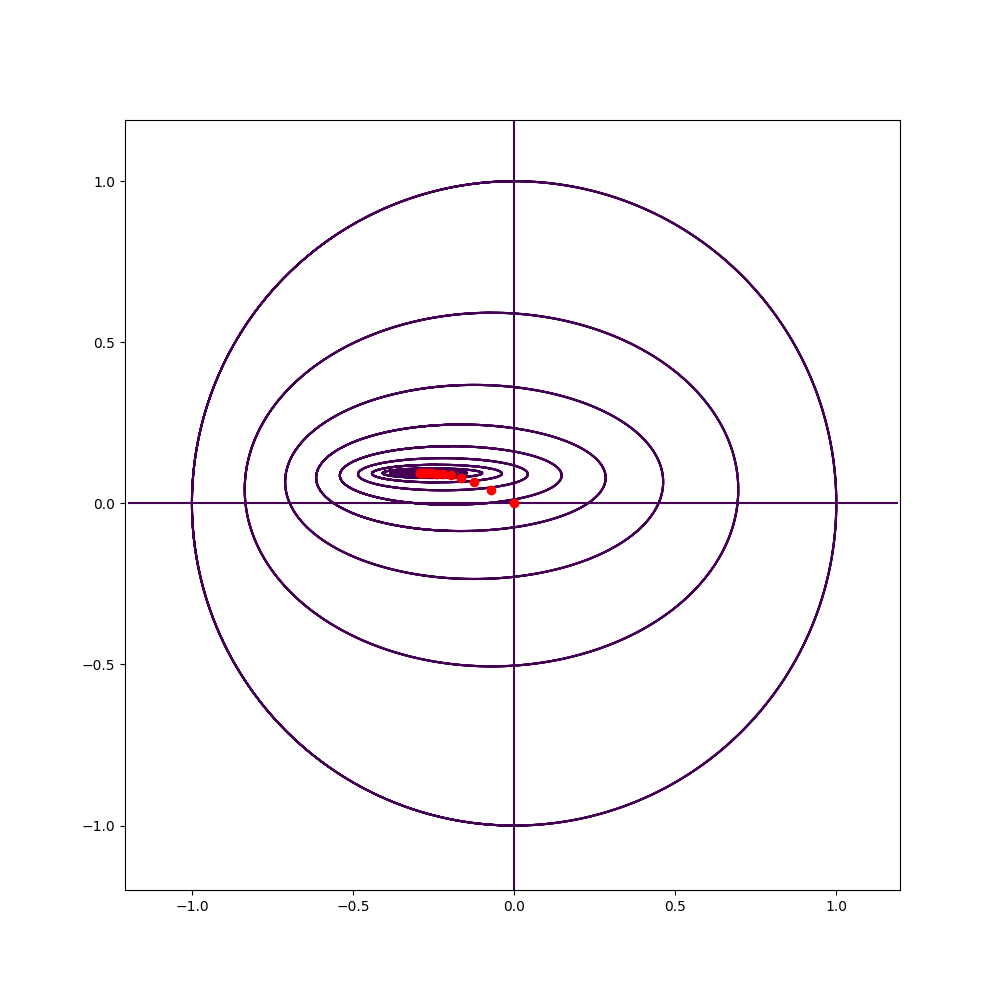}\ \ \ \includegraphics[scale=0.18]{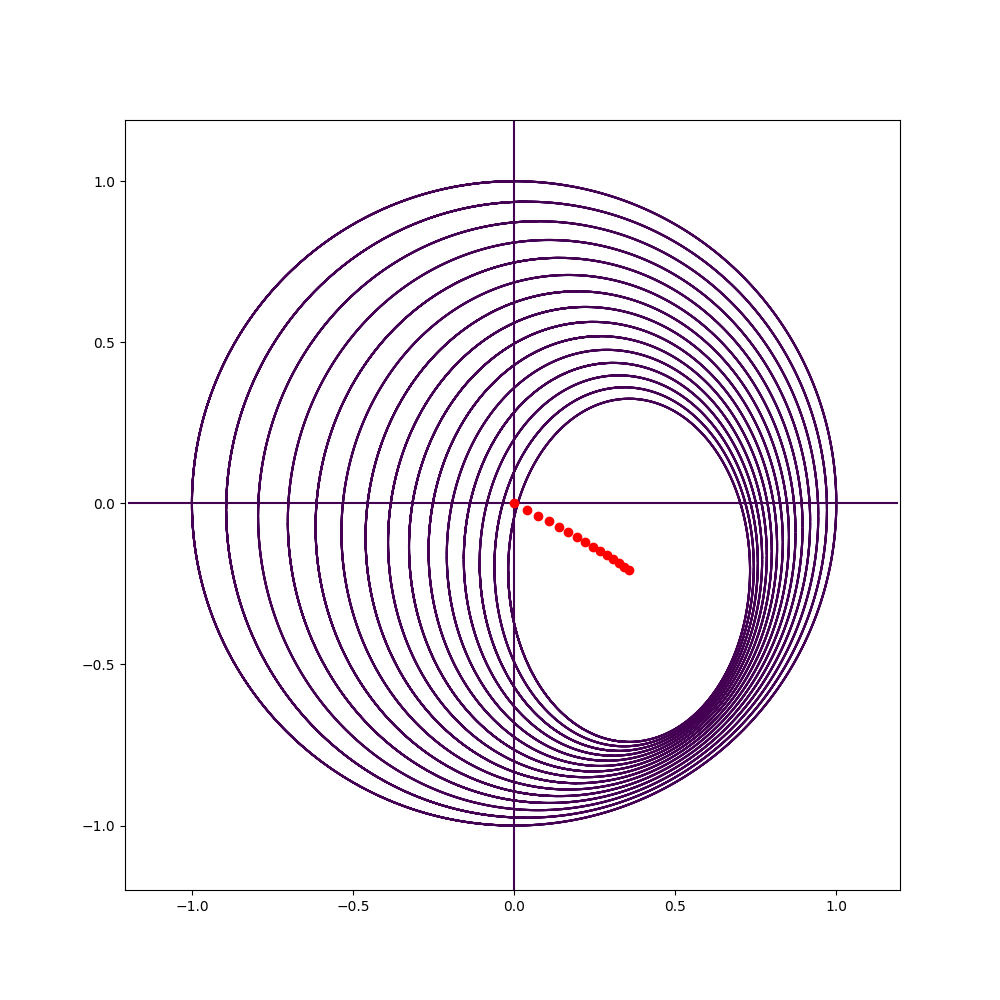} 
	\caption{Dynamics generated by illuminant channels. The time steps are $t = i/10$, $i = 0, \dots, 14$, and the parameters of the Lindblad matrix are as follows: $\alpha_{1,1} = 2$, $\alpha_{2,2} = 1/3$, $\alpha_{3,3} = 1$, $\alpha_{1,3} = 1/7$, and $\alpha_{2,3} = 1/5$ (left), and $\alpha_{1,1} = 1/8$, $\alpha_{2,2} = 1/4$, $\alpha_{3,3} = 1/10$, $\alpha_{1,3} = -1/20$, and $\alpha_{2,3} = -1/10$ (right).} 
	\label{dynLind}
\end{figure}

\subsection{Opponency channels and colour vision deficiencies}

Besides illuminant channels, one may also consider unital Markovian rebit channels whose Lindbladian satisfies $\alpha_{1,3}=\alpha_{2,3}=\alpha_{3,3}=0$, i.e. channels for which $\sigma_3$ does not appear in $\mathcal L$. Since the Pauli matrices encode Hering's chromatic opponency, such channels act exclusively on the chromatic degrees of freedom, modifying the representation of chromatic states in terms of the opponency axes associated with the unique hues. For this reason, these channels may be used to simulate colour perception deficiencies. 

The following definition formalises this notion.

\begin{definition}
	\label{def3.2} An opponency channel is a unital Markovian rebit channel generated by the Lindbladian
	\beq\label{Lindopp}
	\begin{split}
		\mathcal{L}(\rho_{\bf s})&=
		\alpha_{1,1}\mathcal{L}_1(\rho_{\bf s}) + \alpha_{2,2}\mathcal{L}_2(\rho_{\bf s}) + \alpha_{1,2}\widetilde{\mathcal{L}}_{1,2}(\rho_{\bf s})\\
		& =
		(\alpha_{1,2}s_2 -\alpha_{2,2}s_1) \sigma_1+(\alpha_{1,2}s_1 -\alpha_{1,1}s_2)\sigma_2.
	\end{split}
	\eeq
\end{definition}
\noindent An opponency channel that leaves the unique hues fixed is thus given by
\begin{equation}
	\mathcal{L}(\rho_{\bf s})=\alpha_{1,1}\mathcal{L}_1(\rho_{\bf s}) + \alpha_{2,2}\mathcal{L}_2(\rho_{\bf s})= -(\alpha_{2,2}s_1\sigma_1 +\alpha_{1,1}s_2\sigma_2).
\end{equation}
The differential system arising from the GKSL equation with this Lindbladian is
\begin{equation}
	\begin{pmatrix} \dot{s_1} \\ \dot{s_2} \end{pmatrix} =
	\begin{pmatrix} -2\alpha_{2,2} & 2\alpha_{1,2} \\ 2\alpha_{1,2} & -2\alpha_{1,1} \end{pmatrix}
	\begin{pmatrix} s_1(t) \\ s_2(t) \end{pmatrix}.
\end{equation}
The system matrix can be diagonalized, with eigenvalues given by
\begin{equation}\label{eq:eigenvs}
	\lambda_{\pm} = -(\alpha_{2,2} + \alpha_{1,1}) \pm \sqrt{(\alpha_{2,2} - \alpha_{1,1})^2 + 4\alpha_{1,2}^2}.
\end{equation}
Since the Lindblad matrix
\begin{equation}
	\mathcal M=\pmat{\alpha_{1,1} & \alpha_{1,2} &0\\
		\alpha_{1,2} & \alpha_{2,2} & 0\\
		0 & 0 & 0}
\end{equation}
is symmetric positive semi-definite, the two eigenvalues in eq. \eqref{eq:eigenvs} are negative, and the corresponding eigenvectors are orthogonal. For instance, the channel defined by
\begin{equation}
	\dot{\rho}_{\bf s}=\frac{1}{2}\left(\mathcal{L}_1(\rho_{\bf s}) + \mathcal{L}_2(\rho_{\bf s})\right)+
	\frac{1}{4}\widetilde{\mathcal{L}}_{1,2}(\rho_{\bf s})
\end{equation}
is also defined by
\begin{equation}
	\dot{\rho}_{\bf s'}=\frac{3}{4}\mathcal{L}_1'(\rho_{{\bf s'}}) +
	\frac{1}{4}\mathcal{L}_2'(\rho_{{\bf s'}}),
\end{equation}
with 
\begin{gather}
	\rho_{\bf s}=\frac{1}{2}\left(\sigma_0+s_1\sigma_1+s_2\sigma_2\right)=
	\rho_{{\bf s'}}=\frac{1}{2}\left(\sigma_0+s'_1\sigma'_1+s'_2\sigma'_2\right),
\end{gather}
$\sigma'_1=\sigma_1+\sigma_2$, $\sigma'_2=\sigma_1-\sigma_2$, $s'_1=(s_1+s_2)/2$, $s'_2=(s_1-s_2)/2$, and $F_1'=\sigma'_1/2$, $F_2'=\sigma'_2/2$. 

Note that this type of channels does not appear in the classification results of Section \ref{sec:Markovclassification} Indeed, by imposing $R_1(t)=R_2(t)=I_2$ for all $t\geq 0$, those results are valid in the fixed standard Bloch basis $(\sigma_0,\sigma_1,\sigma_2)$ in which the matrix $D(t)$ of the affine part of $\mathcal C(t)$ is diagonal.

An opponency channel that leaves the unique hues fixed is given by
\begin{equation}
	\begin{cases}
		s_1(t) = s_1(0)e^{-2\alpha_{2,2}t}  \\[3mm]
		s_2(t) = s_2(0)e^{-2\alpha_{1,1}t}.
	\end{cases}
\end{equation}

\subsection{Numerical experiments}
\label{subsec3.2}

Here we present illustrative experiments showing how chromatic distortions induced by Markovian rebit dynamics may affect digital images. Our discussion is intentionally limited to a few basic examples, as a comprehensive analysis would require the introduction of advanced colorimetric notions that fall outside the scope of the present work. 

A preliminary observation is in order: at present, we do not possess a faithful numerical model of the chromatic state space of the rebit. Constructing such a model would require dedicated psycho-visual experiments aimed at identifying the four unique hues and the two chromatic opponency mechanisms. These experiments could be carried out by reproducing, with modern quantum optical devices and protocols, classical procedures such as hue-cancellation experiments originally introduced by Jameson and Hurvich more than seventy years ago, see \cite{JamesonHurvichI:55}. 

Such an experimental programme would make it possible to define a novel colour space fully consistent with both Hering's opponency theory and the quantum information-based framework adopted here, thereby avoiding the need to rely on the CIE (Commission International de l'\'Eclairage) colour spaces.
Nevertheless, it is possible to approximate the state space of Hering's rebit starting from the Hue-Chroma-Value (HCV) conical colour space, by suitably reshaping it so as to respect chromatic opponency. A detailed account of this construction is provided in \cite{Berthier:24SPM}.

For the first experiment, we consider the illuminant channel given by equation (\ref{Lindill}), with $\alpha_{1,1}=0.1$, $\alpha_{2,2}=0.25$, and $\alpha_{3,3}=0.15$, the weakly saturated orange illuminant corresponding to the stationary state $(s_{0,1},s_{0,2})=(1/3,0.2).$ Fig. \ref{exp1} illustrates the chromatic distortion induced by this channel. As $t$ grows, we observe a slight shift towards orange, particularly noticeable in the white square of the colour checker. The colours with reds and greens hues are the ones whose saturation decreases most rapidly. This is consistent with the choice of the illuminant and with the choice of $\alpha_{1,1}$ and $\alpha_{2,2}$: since $\alpha_{2,2}$ is greater than $\alpha_{1,1}$, the contraction in the red-green opponency axis is greater than in the yellow-blue opponency axis.
\begin{figure}[h]
	\centering
	\includegraphics[scale=0.55]{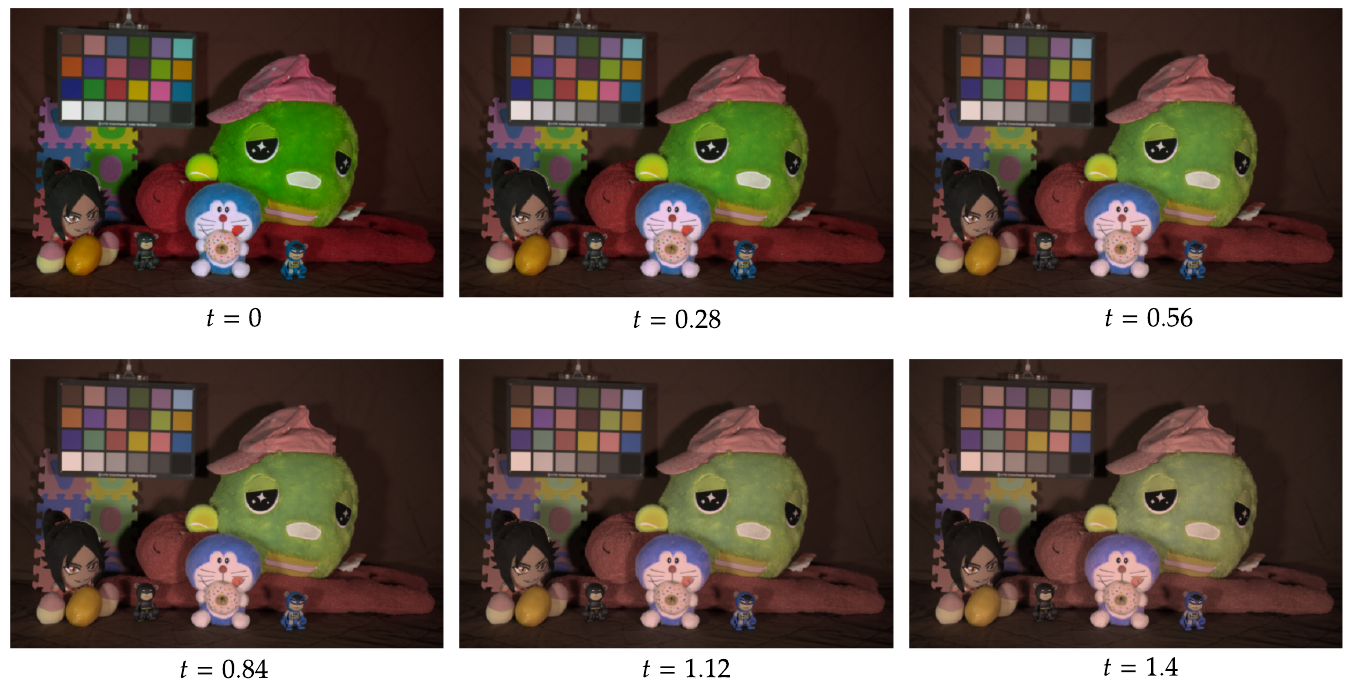}
	\caption{The chromatic distortion induced by an illuminant channel. The weakly saturated orange illuminant is the stationary state $(1/3,0.2)$, $\alpha_{1,1}=0.1$, $\alpha_{2,2}=0.25$, and $\alpha_{3,3}=0.15$.}
	\label{exp1}
\end{figure}

Figure \ref{exp2} shows the chromatic distortion induced by an illuminant channel with a highly saturated green illuminant. As expected, the red patches in the original image, such as the red square of the colour checker and the red support of the green puppet, become first achromatic  and then green as $t$ grows.

\begin{figure}[ht]
	\centering
	\includegraphics[scale=0.46]{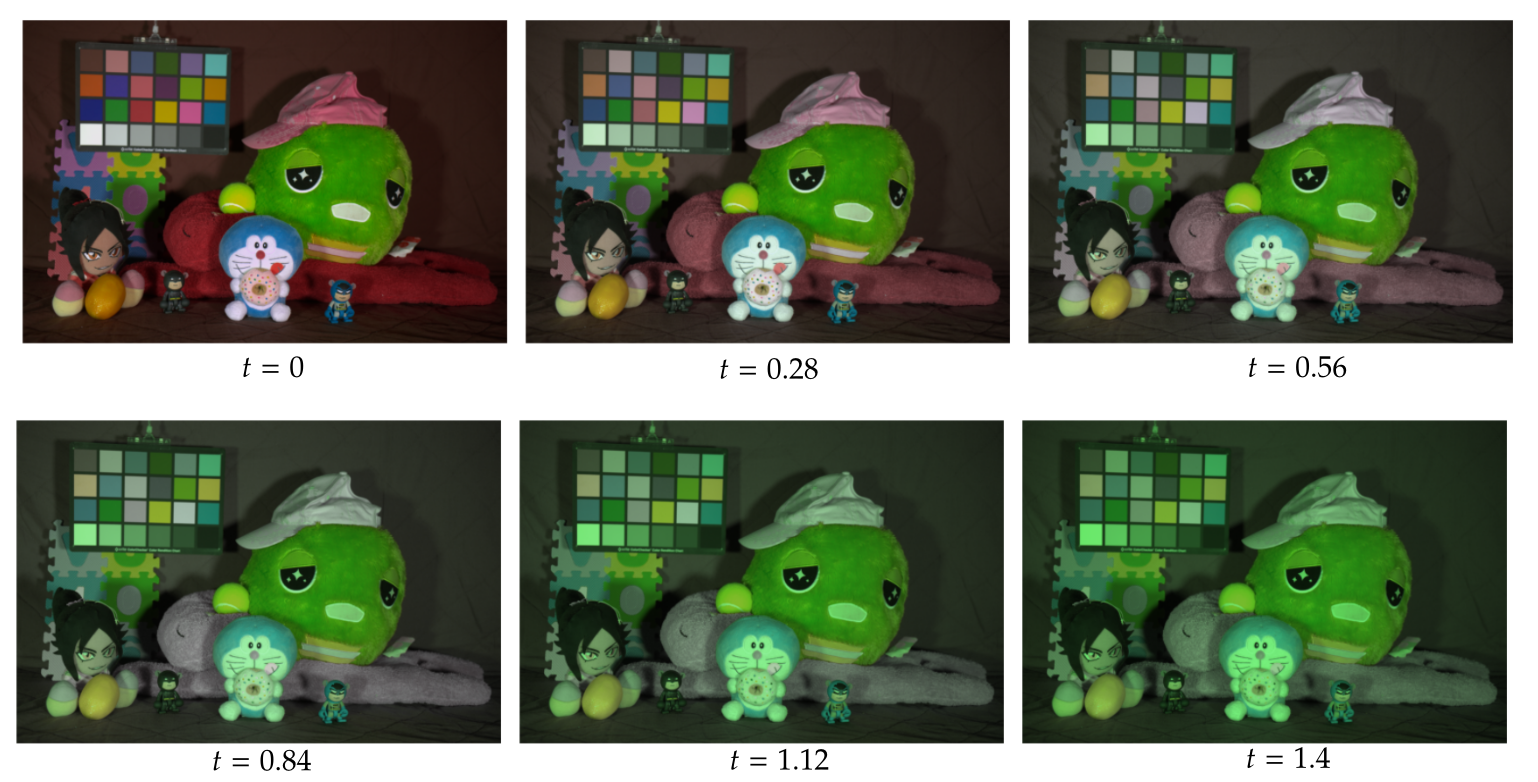}
	\caption{The chromatic distortion induced by an illuminant channel. The strongly saturated green illuminant is the stationary state $(-0.8,0)$.}
	\label{exp2}
\end{figure}

For the last experiment, we consider an opponency channel whose Lindbladian is given by equation (\ref{Lindopp}), with $\alpha_{1,2}=0$, $\alpha_{1,1}=0.005$, and $\alpha_{2,2}=0.4$. This channel aims at reproducing the colour vision of an observer whose perception is affected by a dysfunction in the red-green opponency, i.e. who suffers from red-green colours blindness. Fig. \ref{exp3} shows the chromatic distortion induced by this channel. In particular, the highly saturated colours with hue close to red become achromatic as $t$ grows.

\begin{figure}[ht]
	\centering
	\includegraphics[scale=0.5]{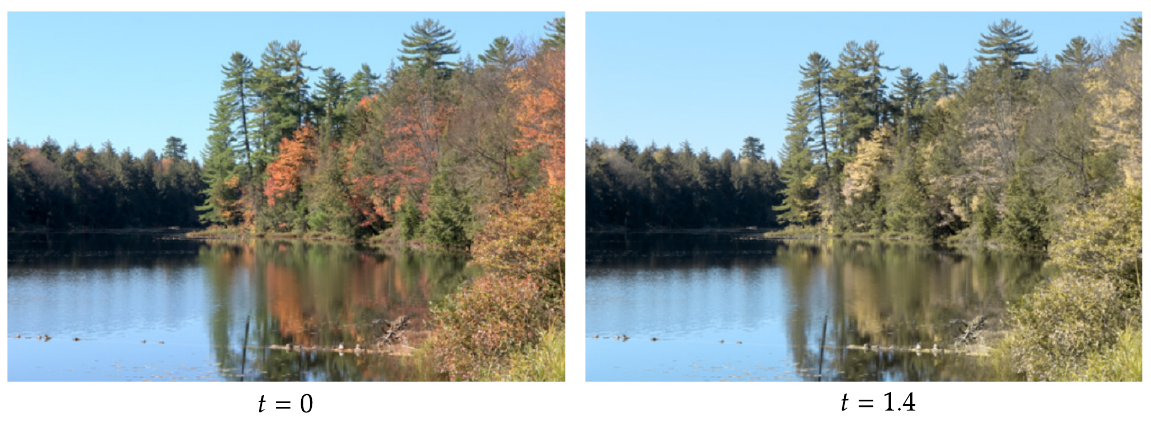}
	\caption{The chromatic distortion induced by an opponency channel with $\alpha_{1,1}=0.005$, and $\alpha_{2,2}=0.4$. }
	\label{exp3}
\end{figure}

\subsection{Relative entropy and chromatic distinguishability}
\label{subsec3.3}
We conclude this section by examining the connection between chromatic distinguishability and the level sets of relative entropy, illustrated in Figure \ref{fig:levelset01} and defined by
\begin{equation}
	\mathcal R_{\mu,{\bf t}} = \left\{ {\vs} \in \mathcal D \; : \; R(\rho_{\bf s} \| \rho_{\bf t}) = \mu \right\}.
\end{equation}



\begin{figure}[ht!]
	\centering
	\includegraphics[scale=0.2]{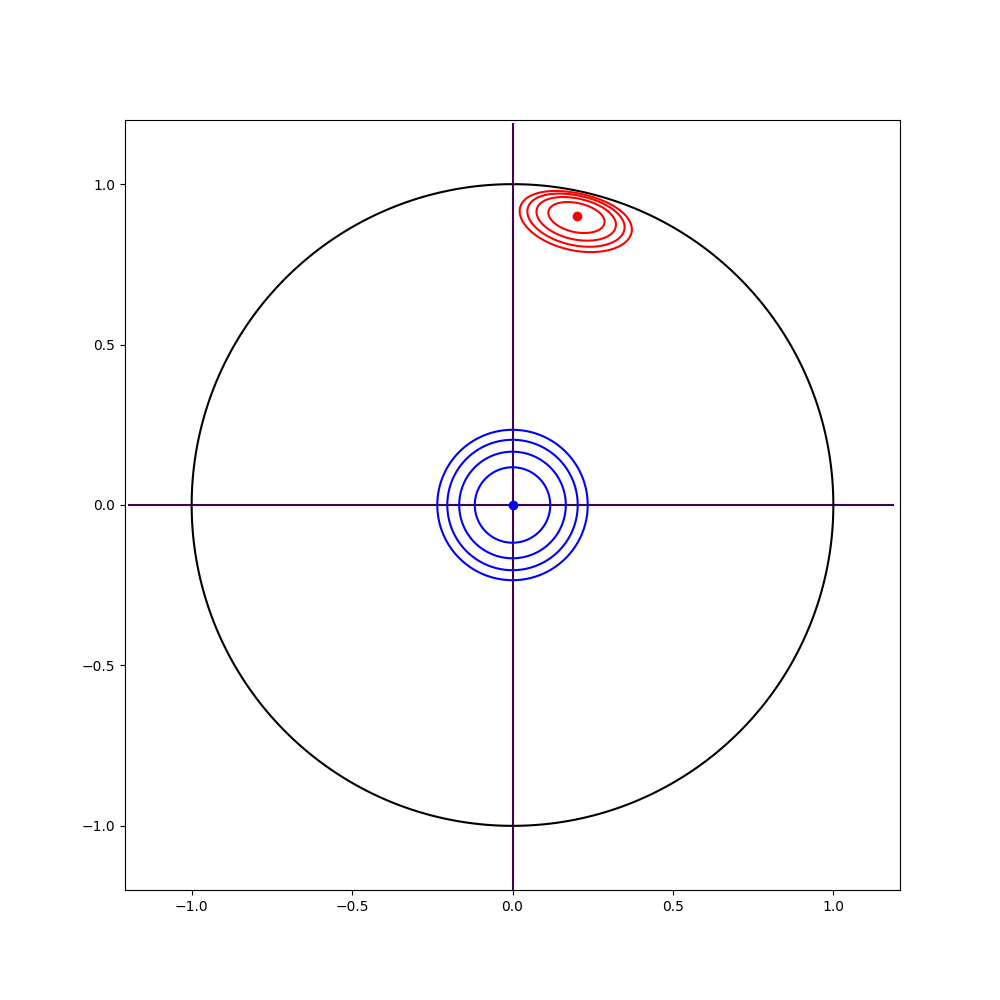}\ \ \ \includegraphics[scale=0.2]{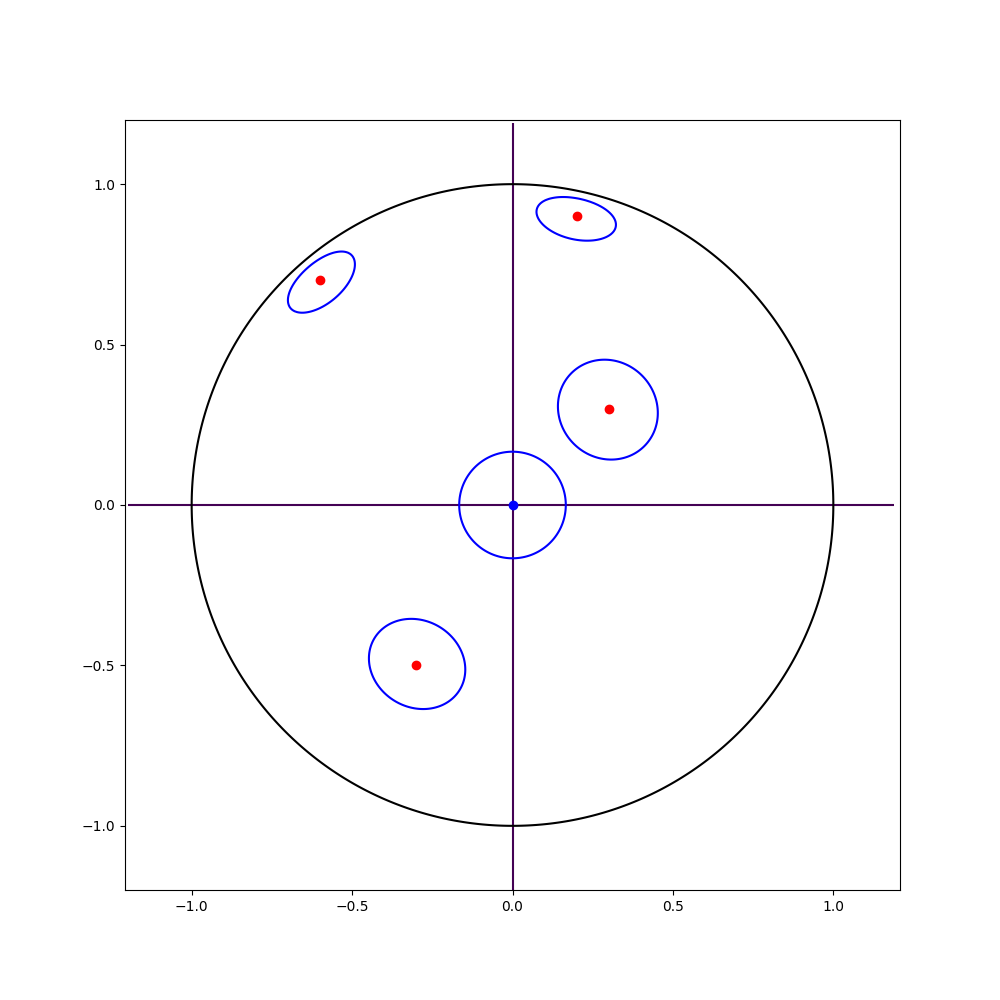}
	\caption{Level sets of relative entropy $\mu=0.01, 0.02, 0.03, 0.04$. For the blue ones ${\bf t}={\bf 0}$, for the red ones ${\bf t}=(0.2, 0.9)$ (left). Level sets of relative entropy $\mu=0.02$ for ${\bf t}=(0,0)$, $(0.2,0.9)$, $(0.3,0.3)$, $(-0.3,-0.5)$, and $(-0.6,0.7)$ (right).}
	\label{fig:levelset01}
\end{figure}

In particular, we show how the presence of a non-achromatic illuminant affects the human ability to discriminate between chromatic states.

For that, consider $\ee_a$, an observer adapted to an achromatic illuminant, two chromatic states $\s$ and $\bf t$, and an illuminant channel $\mathcal C(t)$. Then, by eq. \eqref{eq:rhost}, after the action of $\mathcal C(t)$, $\ee_a$ will perceive the following post-measurement chromatic states
\beq
\varphi^{t}_{\ee_a}(\s)=\rho_{\s(t)}=\mathcal C(t)\rho_\s, \quad \varphi^{t}_{\ee_a}({\bf t})=\rho_{{\bf t}(t)}=\mathcal C(t)\rho_{\bf t}, \quad t\ge 0.
\eeq
Therefore, 
\begin{equation}
	R(\varphi^{t}_{\ee_a}(\s) \| \varphi^{t}_{\ee_a}({\bf t})) = R(\mathcal C(t)\rho_{\s}\|\mathcal C(t)\rho_{\bf t})\le R(\rho_\s\|\rho_{\bf t}),
\end{equation}
where the last inequality follows from the DPI \eqref{eq:DPI}. 
This means that for the observer $\ee_a$ it is more difficult to distinguish chromatic states after the illuminant change.

To illustrate this phenomenon, we consider the level sets
\begin{equation}
	\mathcal{R}^t_\mu = \left\{{\vs} \in \mathcal{D} \;:\; R\left( \rho_{\bf s}(t) \,\|\, \rho_0(t) \right) = \mu \right\},
\end{equation}
and
\begin{equation}
	{\widetilde{\mathcal{R}}}^t_\mu=\left\{{\vs} \in \mathcal{D} \;:\; R\left( \rho_{\bf s} \,\|\, \rho_0(t) \right) = \mu \right\}.
\end{equation}
Figure \ref{fig:levelset04} shows the level sets $\mathcal{R}^t_\mu$, with $t=0$, $t=1.5$, and $\mu=0.01$, in red, and the level set ${\widetilde{\mathcal{R}}}^t_\mu$, with $t = 1.5$ and $\mu=0.01$, in green. The parameters of the illuminant channel are $\alpha_{1,1} = 0.125$, $\alpha_{2,2} = 0.25$, $\alpha_{3,3} = 0.3$, $\alpha_{1,3} = -0.05$, and $\alpha_{2,3} = -0.1.$

\begin{figure}[ht]
	\centering
	\includegraphics[scale=0.25]{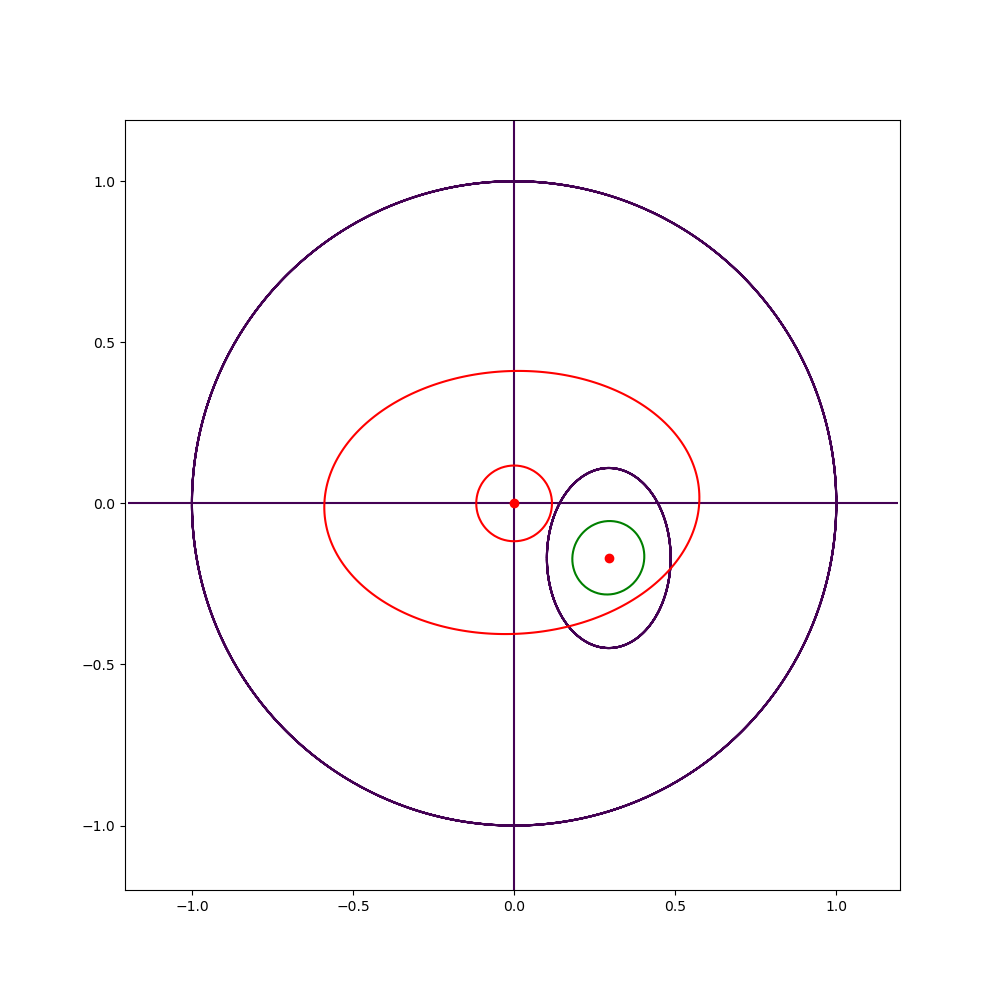}
	\caption{Level sets of the relative entropy for $t = 0$, $t = 1.5$, and $\mu = 0.01$. The parameters of the illuminant channel are: $\alpha_{1,1} = 0.125$, $\alpha_{2,2} = 0.25$, $\alpha_{3,3} = 0.3$, $\alpha_{1,3} = -0.05$, and $\alpha_{2,3} = -0.1.$}\label{fig:levelset04} 
\end{figure}

These results provide a mathematical justification for the observed reduction in colour distinguishability in the lower-right image of, for instance, Figure \ref{exp2}, as compared with the upper-left image.

\section{Conclusions}

We have analysed in detail the structure of Markovian channels of rebit systems. When the affine component of such channels is itself Markovian, their classification turns out to be relatively simple and can be obtained in a direct and explicit manner. By contrast, when the affine component does not form a Markovian semigroup, the classification becomes technically more involved and leads to the emergence of three distinct classes of Markovian channels.

We have also connected the resulting classification to the GKSL master equation adapted to rebit systems, thereby providing a clearer characterisation of the complete positivity conditions.

Finally, we have applied these results to a quantum information-based model of colour perception, in which chromatic states are represented as states of a rebit. Within this setting, we have shown how changes in the illuminant can be modelled by Markovian rebit channels acting on states, giving rise to chromatic distortion and to an inherent reduction of chromatic distinguishability, captured by the monotonicity of quantum relative entropy. We have also seen that other types of channels could be used to study colour vision deficiencies.

In future work, we plan to investigate the complementary phenomenon of chromatic adaptation. Our aim is to describe adaptation as a dynamical map acting on effects rather than on states, counteracting chromatic distortion and thus its associated loss of chromatic distinguishability.

\vskip 1cm

The authors declare no conflict of interest. The paper received no funding. Data availability: does not apply for this contribution.

\vskip 1cm

%

\begin{thebibliography}{31}
	\ifx \bisbn   \undefined \def \bisbn  #1{ISBN #1}\fi
	\ifx \binits  \undefined \def \binits#1{#1}\fi
	\ifx \bauthor  \undefined \def \bauthor#1{#1}\fi
	\ifx \batitle  \undefined \def \batitle#1{#1}\fi
	\ifx \bjtitle  \undefined \def \bjtitle#1{#1}\fi
	\ifx \bvolume  \undefined \def \bvolume#1{\textbf{#1}}\fi
	\ifx \byear  \undefined \def \byear#1{#1}\fi
	\ifx \bissue  \undefined \def \bissue#1{#1}\fi
	\ifx \bfpage  \undefined \def \bfpage#1{#1}\fi
	\ifx \blpage  \undefined \def \blpage #1{#1}\fi
	\ifx \burl  \undefined \def \burl#1{\textsf{#1}}\fi
	\ifx \doiurl  \undefined \def \doiurl#1{\url{https://doi.org/#1}}\fi
	\ifx \betal  \undefined \def \betal{\textit{et al.}}\fi
	\ifx \binstitute  \undefined \def \binstitute#1{#1}\fi
	\ifx \binstitutionaled  \undefined \def \binstitutionaled#1{#1}\fi
	\ifx \bctitle  \undefined \def \bctitle#1{#1}\fi
	\ifx \beditor  \undefined \def \beditor#1{#1}\fi
	\ifx \bpublisher  \undefined \def \bpublisher#1{#1}\fi
	\ifx \bbtitle  \undefined \def \bbtitle#1{#1}\fi
	\ifx \bedition  \undefined \def \bedition#1{#1}\fi
	\ifx \bseriesno  \undefined \def \bseriesno#1{#1}\fi
	\ifx \blocation  \undefined \def \blocation#1{#1}\fi
	\ifx \bsertitle  \undefined \def \bsertitle#1{#1}\fi
	\ifx \bsnm \undefined \def \bsnm#1{#1}\fi
	\ifx \bsuffix \undefined \def \bsuffix#1{#1}\fi
	\ifx \bparticle \undefined \def \bparticle#1{#1}\fi
	\ifx \barticle \undefined \def \barticle#1{#1}\fi
	\bibcommenthead
	\ifx \bconfdate \undefined \def \bconfdate #1{#1}\fi
	\ifx \botherref \undefined \def \botherref #1{#1}\fi
	\ifx \url \undefined \def \url#1{\textsf{#1}}\fi
	\ifx \bchapter \undefined \def \bchapter#1{#1}\fi
	\ifx \bbook \undefined \def \bbook#1{#1}\fi
	\ifx \bcomment \undefined \def \bcomment#1{#1}\fi
	\ifx \oauthor \undefined \def \oauthor#1{#1}\fi
	\ifx \citeauthoryear \undefined \def \citeauthoryear#1{#1}\fi
	\ifx \endbibitem  \undefined \def \endbibitem {}\fi
	\ifx \bconflocation  \undefined \def \bconflocation#1{#1}\fi
	\ifx \arxivurl  \undefined \def \arxivurl#1{\textsf{#1}}\fi
	\csname PreBibitemsHook\endcsname
	
	\bibitem[\protect\citeauthoryear{Wootters}{1990}]{Wootters:1990}
	\begin{barticle}
		\bauthor{\bsnm{Wootters}, \binits{W.K.}}:
		\batitle{Local accessibility of quantum states}.
		\bjtitle{Complexity, entropy and the physics of information}
		\bvolume{8},
		\bfpage{39}--\blpage{46}
		(\byear{1990})
	\end{barticle}
	\endbibitem
	
	\bibitem[\protect\citeauthoryear{McKague et~al.}{2009}]{McKague:2009}
	\begin{barticle}
		\bauthor{\bsnm{McKague}, \binits{M.}},
		\bauthor{\bsnm{Mosca}, \binits{M.}},
		\bauthor{\bsnm{Gisin}, \binits{N.}}:
		\batitle{Simulating quantum systems using real {H}ilbert spaces}.
		\bjtitle{Physical review letters}
		\bvolume{102}(\bissue{2}),
		\bfpage{020505}
		(\byear{2009})
	\end{barticle}
	\endbibitem
	
	\bibitem[\protect\citeauthoryear{Aleksandrova et~al.}{2013}]{Aleksandrova:13}
	\begin{barticle}
		\bauthor{\bsnm{Aleksandrova}, \binits{A.}},
		\bauthor{\bsnm{Borish}, \binits{V.}},
		\bauthor{\bsnm{Wootters}, \binits{W.K.}}:
		\batitle{Real-vector-space quantum theory with a universal quantum bit}.
		\bjtitle{Physical Review A}
		\bvolume{87}(\bissue{5}),
		\bfpage{052106}
		(\byear{2013})
	\end{barticle}
	\endbibitem
	
	\bibitem[\protect\citeauthoryear{Wootters}{2014}]{Wootters:2014}
	\begin{barticle}
		\bauthor{\bsnm{Wootters}, \binits{W.K.}}:
		\batitle{The rebit three-tangle and its relation to two-qubit entanglement}.
		\bjtitle{Journal of Physics A: Mathematical and Theoretical}
		\bvolume{47}(\bissue{42}),
		\bfpage{424037}
		(\byear{2014})
	\end{barticle}
	\endbibitem
	
	\bibitem[\protect\citeauthoryear{Moretti and Oppio}{2017}]{Moretti:2017quantum}
	\begin{barticle}
		\bauthor{\bsnm{Moretti}, \binits{V.}},
		\bauthor{\bsnm{Oppio}, \binits{M.}}:
		\batitle{Quantum theory in real {H}ilbert space: How the complex {H}ilbert
			space structure emerges from {P}oincar{\'e} symmetry}.
		\bjtitle{Reviews in Mathematical Physics}
		\bvolume{29}(\bissue{06}),
		\bfpage{1750021}
		(\byear{2017})
	\end{barticle}
	\endbibitem
	
	\bibitem[\protect\citeauthoryear{Koh et~al.}{2018}]{Koh:2018}
	\begin{barticle}
		\bauthor{\bsnm{Koh}, \binits{D.E.}},
		\bauthor{\bsnm{Niu}, \binits{M.Y.}},
		\bauthor{\bsnm{Yoder}, \binits{T.J.}}:
		\batitle{Quantum simulation from the bottom up: the case of rebits}.
		\bjtitle{Journal of Physics A: Mathematical and Theoretical}
		\bvolume{51}(\bissue{19}),
		\bfpage{195302}
		(\byear{2018})
	\end{barticle}
	\endbibitem
	
	\bibitem[\protect\citeauthoryear{Renou et~al.}{2021}]{Renou:2021}
	\begin{barticle}
		\bauthor{\bsnm{Renou}, \binits{M.-O.}},
		\bauthor{\bsnm{Trillo}, \binits{D.}},
		\bauthor{\bsnm{Weilenmann}, \binits{M.}},
		\bauthor{\bsnm{Le}, \binits{T.P.}},
		\bauthor{\bsnm{Tavakoli}, \binits{A.}},
		\bauthor{\bsnm{Gisin}, \binits{N.}},
		\bauthor{\bsnm{Ac{\'\i}n}, \binits{A.}},
		\bauthor{\bsnm{Navascu{\'e}s}, \binits{M.}}:
		\batitle{Quantum theory based on real numbers can be experimentally falsified}.
		\bjtitle{Nature}
		\bvolume{600}(\bissue{7890}),
		\bfpage{625}--\blpage{629}
		(\byear{2021})
	\end{barticle}
	\endbibitem
	
	\bibitem[\protect\citeauthoryear{Fuchs et~al.}{2022}]{Fuchs:2022}
	\begin{botherref}
		\oauthor{\bsnm{Fuchs}, \binits{C.A.}},
		\oauthor{\bsnm{Olchanyi}, \binits{M.}},
		\oauthor{\bsnm{Weiss}, \binits{M.B.}}:
		Quantum mechanics? it's all fun and games until someone loses an $i$.
		arXiv preprint arXiv:2206.15343
		(2022)
	\end{botherref}
	\endbibitem
	
	\bibitem[\protect\citeauthoryear{Chiribella et~al.}{2023}]{Chiribella:2023}
	\begin{bchapter}
		\bauthor{\bsnm{Chiribella}, \binits{G.}},
		\bauthor{\bsnm{Davidson}, \binits{K.R.}},
		\bauthor{\bsnm{Paulsen}, \binits{V.I.}},
		\bauthor{\bsnm{Rahaman}, \binits{M.}}:
		\bctitle{Positive maps and entanglement in real {H}ilbert spaces}.
		In: \bbtitle{Annales Henri Poincar{\'e}},
		vol. \bseriesno{24},
		pp. \bfpage{4139}--\blpage{4168}
		(\byear{2023}).
		\bcomment{Springer}
	\end{bchapter}
	\endbibitem
	
	\bibitem[\protect\citeauthoryear{Caves et~al.}{2002}]{Caves:2002}
	\begin{barticle}
		\bauthor{\bsnm{Caves}, \binits{C.M.}},
		\bauthor{\bsnm{Fuchs}, \binits{C.A.}},
		\bauthor{\bsnm{Schack}, \binits{R.}}:
		\batitle{Unknown quantum states: The quantum de finetti representation}.
		\bjtitle{Journal of Mathematical Physics}
		\bvolume{43},
		\bfpage{5437}--\blpage{4559}
		(\byear{2002})
	\end{barticle}
	\endbibitem
	
	\bibitem[\protect\citeauthoryear{Caves et~al.}{2001}]{Caves:2001}
	\begin{barticle}
		\bauthor{\bsnm{Caves}, \binits{C.M.}},
		\bauthor{\bsnm{Fuchs}, \binits{C.A.}},
		\bauthor{\bsnm{Rungta}, \binits{P.}}:
		\batitle{Entanglement of formation of an arbitrary state of two rebits}.
		\bjtitle{Foundations of Physics Letters}
		\bvolume{14},
		\bfpage{199}--\blpage{212}
		(\byear{2001})
	\end{barticle}
	\endbibitem
	
	\bibitem[\protect\citeauthoryear{Ald\'e et~al.}{2023}]{Alde:2023}
	\begin{barticle}
		\bauthor{\bsnm{Ald\'e}, \binits{M.}},
		\bauthor{\bsnm{Berthier}, \binits{M.}},
		\bauthor{\bsnm{Provenzi}, \binits{E.}}:
		\batitle{The classification of rebit quantum channels}.
		\bjtitle{Journal of Physics A: Mathematical and Theoretical}
		\bvolume{56}(\bissue{495301}),
		\bfpage{1}--\blpage{17}
		(\byear{2023})
	\end{barticle}
	\endbibitem
	
	\bibitem[\protect\citeauthoryear{Engel and Nagel}{2000}]{Engel:2000}
	\begin{bbook}
		\bauthor{\bsnm{Engel}, \binits{K.-J.}},
		\bauthor{\bsnm{Nagel}, \binits{R.}}:
		\bbtitle{One-parameter Semigroups for Linear Evolution Equations}.
		\bpublisher{Springer},
		\blocation{New York}
		(\byear{2000})
	\end{bbook}
	\endbibitem
	
	\bibitem[\protect\citeauthoryear{Heinosaari and Ziman}{2011}]{Heinosaari:2011}
	\begin{bbook}
		\bauthor{\bsnm{Heinosaari}, \binits{T.}},
		\bauthor{\bsnm{Ziman}, \binits{M.}}:
		\bbtitle{The Mathematical Language of Quantum Theory: from Uncertainty to
			Entanglement}.
		\bpublisher{Cambridge University Press},
		\blocation{UK}
		(\byear{2011})
	\end{bbook}
	\endbibitem
	
	\bibitem[\protect\citeauthoryear{Rubilar and Schultz}{2020}]{Rubilar:2020}
	\begin{barticle}
		\bauthor{\bsnm{Rubilar}, \binits{F.}},
		\bauthor{\bsnm{Schultz}, \binits{L.}}:
		\batitle{Adjoint orbits of sl (2, r) and their geometry}.
		\bjtitle{Pro Mathematica}
		\bvolume{31}(\bissue{61}),
		\bfpage{73}--\blpage{107}
		(\byear{2020})
	\end{barticle}
	\endbibitem
	
	\bibitem[\protect\citeauthoryear{Breuer and Petruccione}{2002}]{Breuer:2002}
	\begin{bbook}
		\bauthor{\bsnm{Breuer}, \binits{H.-P.}},
		\bauthor{\bsnm{Petruccione}, \binits{F.}}:
		\bbtitle{The Theory of Open Quantum Systems}.
		\bpublisher{Oxford University Press},
		\blocation{USA}
		(\byear{2002})
	\end{bbook}
	\endbibitem
	
	\bibitem[\protect\citeauthoryear{Gorini et~al.}{1976}]{Gorini:76}
	\begin{barticle}
		\bauthor{\bsnm{Gorini}, \binits{V.}},
		\bauthor{\bsnm{Kossakowski}, \binits{A.}},
		\bauthor{\bsnm{Sudarshan}, \binits{E.C.G.}}:
		\batitle{Completely positive dynamical semigroups of n-level systems}.
		\bjtitle{Journal of Mathematical Physics}
		\bvolume{17}(\bissue{5}),
		\bfpage{821}--\blpage{825}
		(\byear{1976})
	\end{barticle}
	\endbibitem
	
	\bibitem[\protect\citeauthoryear{Lindblad}{1976}]{Lindblad:76}
	\begin{barticle}
		\bauthor{\bsnm{Lindblad}, \binits{G.}}:
		\batitle{On the generators of quantum dynamical semigroups}.
		\bjtitle{Communications in Mathematical Physics}
		\bvolume{48}(\bissue{2}),
		\bfpage{119}--\blpage{130}
		(\byear{1976})
	\end{barticle}
	\endbibitem
	
	\bibitem[\protect\citeauthoryear{Baez}{2012}]{Baez:12}
	\begin{barticle}
		\bauthor{\bsnm{Baez}, \binits{J.C.}}:
		\batitle{Division algebras and quantum theory}.
		\bjtitle{Foundations of Physics}
		\bvolume{42}(\bissue{7}),
		\bfpage{819}--\blpage{855}
		(\byear{2012})
	\end{barticle}
	\endbibitem
	
	\bibitem[\protect\citeauthoryear{Yao et~al.}{2019}]{Yao:2019}
	\begin{botherref}
		\oauthor{\bsnm{Yao}, \binits{C.}},
		\oauthor{\bsnm{Ma}, \binits{Y.}},
		\oauthor{\bsnm{Chen}, \binits{L.}}:
		Derivations on semi-simple {J}ordan algebras and its applications.
		ar{X}iv preprint ar{X}iv:1906.04552
		(2019)
	\end{botherref}
	\endbibitem
	
	\bibitem[\protect\citeauthoryear{Berthier and
		Provenzi}{2022}]{BerthierProvenzi:2022PRS}
	\begin{barticle}
		\bauthor{\bsnm{Berthier}, \binits{M.}},
		\bauthor{\bsnm{Provenzi}, \binits{E.}}:
		\batitle{Quantum measurement and colour perception: theory and applications}.
		\bjtitle{Proceedings of the Royal Society A}
		\bvolume{478}(\bissue{2258}),
		\bfpage{20210508}
		(\byear{2022})
	\end{barticle}
	\endbibitem
	
	\bibitem[\protect\citeauthoryear{Berthier et~al.}{2022}]{Berthier:22SIAM}
	\begin{barticle}
		\bauthor{\bsnm{Berthier}, \binits{M.}},
		\bauthor{\bsnm{Prencipe}, \binits{N.}},
		\bauthor{\bsnm{Provenzi}, \binits{E.}}:
		\batitle{A quantum information-based refoundation of color perception
			concepts}.
		\bjtitle{SIAM Journal on Imaging Sciences}
		\bvolume{15}(\bissue{4}),
		\bfpage{1944}--\blpage{1976}
		(\byear{2022})
	\end{barticle}
	\endbibitem
	
	\bibitem[\protect\citeauthoryear{Berthier}{2020}]{Berthier:2020}
	\begin{barticle}
		\bauthor{\bsnm{Berthier}, \binits{M.}}:
		\batitle{Geometry of color perception. {P}art 2: perceived colors from real
			quantum states and {H}ering?s rebit}.
		\bjtitle{The Journal of Mathematical Neuroscience}
		\bvolume{10}(\bissue{1}),
		\bfpage{1}--\blpage{25}
		(\byear{2020})
	\end{barticle}
	\endbibitem
	
	\bibitem[\protect\citeauthoryear{Kraus et~al.}{1983}]{Kraus:83}
	\begin{botherref}
		\oauthor{\bsnm{Kraus}, \binits{K.}},
		\oauthor{\bsnm{B{\"o}hm}, \binits{A.}},
		\oauthor{\bsnm{Dollard}, \binits{J.D.}},
		\oauthor{\bsnm{Wootters}, \binits{W.}}:
		States, effects, and operations: fundamental notions of quantum theory.
		lectures in mathematical physics at the university of texas at austin.
		Lecture notes in physics
		\textbf{190}
		(1983)
	\end{botherref}
	\endbibitem
	
	\bibitem[\protect\citeauthoryear{Busch et~al.}{1997}]{Busch:97}
	\begin{bbook}
		\bauthor{\bsnm{Busch}, \binits{P.}},
		\bauthor{\bsnm{Grabowski}, \binits{M.}},
		\bauthor{\bsnm{Lahti}, \binits{P.J.}}:
		\bbtitle{Operational Quantum Physics}
		vol. \bseriesno{31}.
		\bpublisher{Springer},
		\blocation{Berlin}
		(\byear{1997})
	\end{bbook}
	\endbibitem
	
	\bibitem[\protect\citeauthoryear{Berthier et~al.}{2024}]{Berthier:24SPM}
	\begin{barticle}
		\bauthor{\bsnm{Berthier}, \binits{M.}},
		\bauthor{\bsnm{Prencipe}, \binits{N.}},
		\bauthor{\bsnm{Provenzi}, \binits{E.}}:
		\batitle{Split-quaternions for perceptual white balance}.
		\bjtitle{IEEE Signal Processing Magazine}
		\bvolume{41}(\bissue{4}),
		\bfpage{42}--\blpage{50}
		(\byear{2024})
	\end{barticle}
	\endbibitem
	
	\bibitem[\protect\citeauthoryear{Cortese}{2002}]{Cortese:02}
	\begin{botherref}
		\oauthor{\bsnm{Cortese}, \binits{J.}}:
		Relative entropy and single qubit {H}olevo-{S}chumacher-{W}estmoreland channel
		capacity.
		arXiv: Quantum Physics
		(2002)
	\end{botherref}
	\endbibitem
	
	\bibitem[\protect\citeauthoryear{Wilde}{2017}]{Wilde:2017}
	\begin{bbook}
		\bauthor{\bsnm{Wilde}, \binits{M.M.}}:
		\bbtitle{Quantum Information Theory}.
		\bpublisher{Cambridge University Press},
		\blocation{Cambridge, UK}
		(\byear{2017})
	\end{bbook}
	\endbibitem
	
	\bibitem[\protect\citeauthoryear{Witten}{2020}]{Witten:2020}
	\begin{barticle}
		\bauthor{\bsnm{Witten}, \binits{E.}}:
		\batitle{A mini-introduction to information theory}.
		\bjtitle{La Rivista del Nuovo Cimento}
		\bvolume{43},
		\bfpage{187}--\blpage{227}
		(\byear{2020})
	\end{barticle}
	\endbibitem
	
	\bibitem[\protect\citeauthoryear{Auletta et~al.}{2009}]{Auletta:09}
	\begin{bbook}
		\bauthor{\bsnm{Auletta}, \binits{G.}},
		\bauthor{\bsnm{Fortunato}, \binits{M.}},
		\bauthor{\bsnm{Parisi}, \binits{G.}}:
		\bbtitle{Quantum Mechanics}.
		\bpublisher{Cambridge University Press},
		\blocation{UK}
		(\byear{2009})
	\end{bbook}
	\endbibitem
	
	\bibitem[\protect\citeauthoryear{Hurvich and
		Jameson}{1955}]{JamesonHurvichI:55}
	\begin{barticle}
		\bauthor{\bsnm{Hurvich}, \binits{L.M.}},
		\bauthor{\bsnm{Jameson}, \binits{D.}}:
		\batitle{Some quantitative aspects of an opponent-colors theory. {I.} chromatic
			responses and spectral saturation}.
		\bjtitle{JOSA}
		\bvolume{45}(\bissue{7}),
		\bfpage{546}--\blpage{552}
		(\byear{1955})
	\end{barticle}
	\endbibitem
	
\end{thebibliography}


\end{document}